\documentclass[prl,twocolumn,floatfix,superscriptaddress,longbibliography, nofootinbib]{revtex4-2}
\usepackage{
    adjustbox,
    algorithm,
    algpseudocode,
    amsmath,
    amssymb,
    amsthm,
    booktabs,
    braket,
    comment,
    csquotes,
    enumitem,
    graphicx,
    IEEEtrantools,
    mathtools,
    multirow,
    natbib,
    physics,
    refcount,
    relsize,
    url,
    times,
    xcolor
}

\usepackage[english]{babel}
\usepackage[normalem]{ulem}
\usepackage[caption=false]{subfig}
\usepackage{pifont}
\usepackage{dsfont}

\definecolor{mainblue}{HTML}{1f77b4}
\definecolor{mainorange}{HTML}{ff7f0e}
\definecolor{maingreen}{HTML}{2ca02c}
\definecolor{mainred}{HTML}{DC3522}
\definecolor{mainpurple}{HTML}{9467bd}
\definecolor{mainpink}{HTML}{e377c2}

\usepackage[dvipsnames]{xcolor}
\usepackage[colorlinks=true, urlcolor=Maroon, linkcolor=Emerald, citecolor=Periwinkle]{hyperref}

\DeclareMathOperator*{\argmax}{argmax}

\newcommand{\pder}[2][]{\frac{\partial#1}{\partial#2}}

\providecommand{\hC}{\ensuremath{\hat{C}}}

\providecommand{\hE}{\ensuremath{\hat{E}}}

\providecommand{\hH}{\ensuremath{\hat{H}}}

\providecommand{\hvartheta}{\ensuremath{\hat{\vartheta}}}

\providecommand{\hcalL}{\ensuremath{\hat{\mathcal{L}}}}

\providecommand{\tC}{\ensuremath{\Tilde{C}}}

\providecommand{\tZ}{\ensuremath{\Tilde{Z}}}

\providecommand{\tg}{\ensuremath{\Tilde{g}}}

\providecommand{\to}{\ensuremath{\Tilde{o}}}

\providecommand{\calD}{\ensuremath{\mathcal{D}}}

\providecommand{\calF}{\ensuremath{\mathcal{F}}}

\providecommand{\calL}{\ensuremath{\mathcal{L}}}

\providecommand{\calN}{\ensuremath{\mathcal{N}}}
\providecommand{\calO}{\ensuremath{\mathcal{O}}}
\providecommand{\calP}{\ensuremath{\mathcal{P}}}

\providecommand{\calX}{\ensuremath{\mathcal{X}}}
\providecommand{\calY}{\ensuremath{\mathcal{Y}}}
\providecommand{\calZ}{\ensuremath{\mathcal{Z}}}

\DeclareMathOperator*{\bbE}{\ensuremath{\mathbb{E}}}

\providecommand{\bbI}{\ensuremath{\mathbb{I}}}

\DeclareMathOperator*{\bbP}{\ensuremath{\mathbb{P}}}

\providecommand{\bbR}{\ensuremath{\mathbb{R}}}

\newtheorem{theorem}{Theorem}
\newtheorem{lemma}[theorem]{Lemma}

\newtheorem{definition}[theorem]{Definition}
\newtheorem{corollary}[theorem]{Corollary}
\newtheorem{conjecture}[theorem]{Conjecture}
\newtheorem{assumption}{Assumption}

\newenvironment{remark}[1][Remark]{\begin{trivlist}
\item[\hskip \labelsep {\bfseries #1}]}{\end{trivlist}}

\newtheorem{innercustomgeneric}{\customgenericname}
\providecommand{\customgenericname}{}
\newcommand{\newcustomtheorem}[2]{%
  \newenvironment{#1}[1]
  {%
   \renewcommand\customgenericname{#2}%
   \renewcommand\theinnercustomgeneric{##1}%
   \innercustomgeneric
  }
  {\endinnercustomgeneric}
}

\newcustomtheorem{customtheorem}{Theorem}
\newcustomtheorem{customlemma}{Lemma}
\newcustomtheorem{customproposition}{Proposition}
\newcustomtheorem{customassumption}{Assumption}

\newcommand{\fu}{Dahlem Center for Complex Quantum Systems, Freie Universit\"{a}t Berlin, 14195 Berlin, Germany}

\newcommand{\hhi}{Fraunhofer Heinrich Hertz Institute, 10587 Berlin, Germany}
\newcommand{\leiden}{Leiden University, Niels Bohrweg 1, 2333 CA Leiden, Netherlands}
\newcommand{\vw}{Volkswagen Group Innovation, Berliner Ring 2, 38440 Wolfsburg, Germany}
\newcommand{\physwaterloo}{Department of Physics and Astronomy, University of Waterloo, ON N2L 3G1, Canada}
\newcommand{\vinst}{Vector Institute, Toronto, ON M5G 0C6, Canada}
\newcommand{\porsche}{Porsche Digital GmbH, 71636 Ludwigsburg, Germany}

\begin{document}

\title{Cautious optimism for deep parameterized quantum circuits}

\author{Marie Kempkes}
\email{marie.kempkes@volkswagen.de}
\affiliation{\leiden}
\affiliation{\vw}
\author{Elies Gil-Fuster}
\affiliation{\fu}
\affiliation{\hhi}
\author{Carlos Bravo-Prieto}
\affiliation{\fu}
\author{Aroosa Ijaz}
\affiliation{\physwaterloo}
\affiliation{\vinst}
\author{Alissa Wilms}
\affiliation{\fu}
\affiliation{\porsche}
\author{Jens Eisert}
\affiliation{\fu}
\author{Evert van Nieuwenburg}
\affiliation{\leiden}
\author{Vedran Dunjko}
\affiliation{\leiden}

\begin{abstract}
    A central challenge in quantum machine learning is understanding the scaling behavior of \emph{parameterized quantum circuits} (PQCs).
    In particular, it remains unclear how their performance on unseen data changes as the number of trainable parameters increases.
    Prior works have derived formal generalization guarantees for quantum models, but it is well-known that many such results do not fully characterize generalization behavior in practice.
    In this work, we show that gradient-based PQCs can exhibit improved performance on unseen data as model size increases, displaying the phenomenon of double descent.
    This contrasts with the traditional view that larger models lead to degraded generalization.
    We provide analytical results rigorously underpinning this behavior by leveraging add-one-in perturbation techniques and spectral properties of random matrices.
    We support these results with numerical experiments on re-uploading PQCs across several data sets and training set sizes, consistently observing the predicted double descent behavior.
    While other obstacles on the path toward practical quantum machine learning remain, our finding that deeper parameterized quantum circuits do not necessarily exhibit degraded performance provides reasons for cautious optimism.
\end{abstract}

\maketitle

    \emph{Quantum machine learning} (QML) has emerged as a promising paradigm at the intersection of quantum computing and data-driven learning, with the potential to outperform classical methods for specific learning tasks.
    Indeed, a growing body of work has identified concrete settings in which quantum models can offer computational advantages over classical counterparts~\cite{QuantumFeatureEmbedding2,PACLearning,DensityModelling,Molteni,VedranExponential,RevModPhys.91.045002,biamonte2017quantum,masot2025prospects,barthe2025quantum,bandyopadhyay2026provable,MindTheGaps}.
    Despite this substantial progress, many fundamental aspects of QML remain poorly understood. In particular, it is still unclear whether, and under what precise conditions, quantum advantage can be expected, especially as models are scaled to larger numbers of trainable parameters.
    While several limitations to the trainability of \emph{parameterized quantum circuits}  (PQCs) are now well characterized~\cite{BarrenPlateaus,cerezo2021cost,NoiseInduced,AnschuetzTraps,mele2026noise}, considerable effort has also been devoted to understanding their generalization behavior through bounds on the difference between empirical and expected risk.
    These bounds depend, for example, on the training set size, the number of trainable gates, and the data-encoding strategy~\cite{caro2021encoding,abbas2021power,Banchi2021,caro2022generalization,du2022efficient,UnderstandingQML,rodriguez2026pac, peters2023generalization}.
    However, they do not directly address how the expected risk attained by a trained model changes as the number of parameters crosses the so-called \emph{interpolation threshold}, namely, the critical point at which the number of parameters becomes comparable to the number of training points.   

    In classical machine learning, studying the behavior of models across this threshold has led to substantial insights through the so-called \emph{double descent} phenomenon~\cite{belkin2019reconciling}.
    Contrary to traditional expectations from statistical learning theory, double descent describes a non-monotonic relationship between test 
    error and model size.
    As the number of parameters increases, the test error initially follows a classical U-shaped curve, peaking near the interpolation threshold.
    Remarkably, further overparameterization often leads to a second descent phase, during which the test error decreases again and eventually saturates.
    This behavior marks a transition from the underparameterized to the overparameterized regime and has been observed across a wide range of classical models and learning tasks~\cite{belkin2019reconciling,nakkiran2021deep,hastie2022surprises,DoubleDescent2}.

    In the realm of quantum machine learning, however, analogous scaling behavior has so far been much less explored.
    Ref.~\cite{kempkes2025double} has shown that quantum kernel methods can, in principle, exhibit a similar scaling behavior in the overparameterized regime, assuming trainability can be maintained.
    More recently, Ref.~\cite{pranjic2026grokking} has reported non-monotonic test-error dynamics during the training of overparameterized PQCs.
    This naturally raises the question of whether PQCs trained with gradient-based methods also exhibit double descent as a function of the number of trainable parameters.
    To date, some evidence has pointed toward a standard U-shaped risk curve for PQCs, suggesting that overparameterization may not yield a second descent~\cite{du2023problem}. 
    This conclusion is obtained for circuits in which concentration effects dominate, a regime often associated with poor trainability.

    In this work, we instead focus on PQCs operating in a trainable regime and study their generalization properties as the number of trainable parameters increases.
    Building on an existing lower bound on the expected risk, we show that, under reasonable assumptions, the bound attains a maximum at the interpolation threshold.
    We further derive a corresponding upper bound and show that it likewise attains a maximum at interpolation.
    Together, these results establish a double descent behavior in the risk bounds of PQCs analogous to that observed in classical machine learning.
    We further support our theoretical findings with numerical experiments on a variety of datasets.
    Overall, our results provide new insight into the scaling behavior of gradient-based quantum models and suggest that overparameterization may play a role in PQCs similar to that observed in classical machine learning, offering cautious optimism for the prospects of quantum machine learning at scale.

\section{Preliminaries}

    We consider the ubiquitous setting of supervised learning with quantum learning models based on \emph{parametrized quantum circuits} (PQCs).
    Refer to Appendix~\ref{a:framework} for a fully-formal presentation of our framework.
    We consider a learning task on $d$-dimensional inputs $x$ and $K$-dimensional labels $y$, and assume we are given a training set $S=\{(x_i, y_i)\}_{i=1}^N$ of $N$ input-output pairs, sampled i.i.d.\ from the underlying distribution that defines the problem.
    
    Consider \emph{data re-uploading} PQCs~\cite{perez2020data}, where data-dependent layers $U_\ell(x)$ and trainable layers $V_\ell(\vartheta)$ are alternated, for $\ell\in\{1,\ldots,L\}$.
    Here, $\vartheta$ is a $p$-dimensional vector of trainable parameters.
    Starting from an initial state, the data re-uploading PQC prepares a parametrized state $\rho(x;\vartheta)$ that depends on both the data $x$ and the parameters $\vartheta$.
    To fully specify our model, we fix observables $O_k$ for $k\in\{1,\ldots,K\}$, from which we define the hypothesis functions.
    Specifically, the $k^\text{th}$ output of the hypothesis function is given by the expectation value of $O_k$ with respect to the parametrized state $\rho(x;\vartheta)$
    \begin{align}
        [f_\vartheta(x)]_k = \Tr\left\{\rho(x;\vartheta)O_k\right\}.
    \end{align}

    Given a loss function $\ell((x,y);\vartheta)$, our goal is to minimize the \emph{expected risk}, defined as the average loss over the problem distribution
    \begin{align}
    \calL(\vartheta)=\bbE[\ell((x,y);\vartheta)].
    \end{align}
    In practice, we assume the underlying distribution is unknown, and we explicitly optimize the \emph{empirical risk} -- the average loss over the training set 
    \begin{align}
        \hcalL_S(\vartheta) = \frac{1}{N}\sum_{i=1}^N \ell((x_i,y_i);\vartheta),
    \end{align}
    where $(x_i,y_i)$ are the elements of the training set $S$.
    In this work, we focus on the mean-squared error loss $\ell((x,y);\vartheta) = \frac{1}{2}\lVert f_\vartheta(x)-y\rVert^2$.

    We next follow the formalism in Ref.~\cite{singh2022phenomenology}, with further technical details being deferred to Appendix~\ref{a:expected_risk_local_minima}.
    Our study is specialized to the final parameters reached during training, which we denote by $\hvartheta_S$.
    We make the explicit assumption that $\hvartheta_S$ is a \emph{local minimum of the empirical risk}, as is common in gradient-based QML.
    Under this assumption, $\hvartheta_S$ satisfies
    \begin{align}
    \nabla_\vartheta \hcalL_S(\hvartheta_S) = 0
    \qquad\text{and}\qquad
    \hH_S(\hvartheta_S) \succeq 0,
    \end{align}
    where $\hH_S(\vartheta)\coloneqq \nabla_\vartheta^2\hcalL_S(\vartheta)$ is the Hessian matrix of the empirical risk.

    We now derive an approximation of the expected risk via a so-called \emph{add-one-in} analysis.
    Specifically, we consider inserting an additional pair $(x',y')$ to the training set $S$.
    Denote by $\hvartheta_{S'}$ the final training parameters, where $S'$ is the union of $S$ and the new datum.
    The value of the loss $\ell((x',y');\hvartheta_{S'})$ will depend on the new datum.
    With these, define the so-called \emph{add-one-in loss} $\calL'(S)$ as
    \begin{align}
    \calL'(S) \coloneqq  \mathbb{E}\left[\ell\left((x',y');\hat{\vartheta}_{S'}\right)\right],
    \end{align}
    where the average is over $(x',y')$ being drawn from the problem distribution.
    While not immediately obvious, the quantity $\calL'(S)$ plays a central role in deriving an approximation of the expected risk.
    The key idea is that the transition from $S$ to $S'$ constitutes a small perturbation of the training set, allowing the parameter displacement $\hvartheta_{S'}-\hvartheta_S$ and the associated change in loss at $(x', y')$ to be approximated via so-called influence functions~\cite{hampel1986robust}.
    Averaging over the choice of the additional sample $(x',y')$ then yields an approximation of the expected risk that is accurate up to corrections of order $\mathcal O(N^{-2})$.

    To state the resulting expression, we introduce the gradients of the loss function $J_\ell((x,y);\vartheta)=\nabla_\vartheta\ell((x,y);\vartheta)$ and the gradients of the hypothesis functions $J_f((x,y);\vartheta)=\nabla_\vartheta f_\vartheta(x)$.
    Denote by
    \begin{align}
    C(\vartheta) \coloneqq \mathbb{E} \left[ J_\ell((x,y);\vartheta) J_\ell((x,y);\vartheta)^\intercal \right]
    \end{align}
    the average of the outer product of the loss gradients over the problem distribution.
    Analogously, $C_f(\vartheta)$ is given by the average of the outer product of the hypothesis gradients
    \begin{align}
        C_f(\vartheta) \coloneqq     \mathbb{E}    \left[J_f((x,y);\vartheta)J_f((x,y);\vartheta)^\intercal \right].
    \end{align}
    As a technical comment, we condition the average in $C_f(\vartheta)$ to be only over samples $(x,y)$ for which $\lVert J_f((x,y);\vartheta)\rVert^2$ is greater than zero.
    Note that both $C(\vartheta)$ and $C_f(\vartheta)$ are uncentered covariance matrices.

    With these quantities, Theorem 3 in Ref.~\cite{singh2022phenomenology} provides an error decomposition for the expected risk 
    \begin{align}
        \calL(\hvartheta_S) = \calL'(S) +& \frac{1}{N+1}\Tr\left\{\left[\hH_S(\hvartheta_S)\right]^{-1}C(\hvartheta_S)\right\}
    \end{align}
    evaluated at the local minimum $\hvartheta_S$,
    where the inverse of the Hessian is understood as the Moore--Penrose pseudoinverse. The decomposition holds up to second- and higher-order corrections $\calO\left(N^{-2}\right)$.
    Building on this decomposition, Theorem 4 in Ref.~\cite{singh2022phenomenology} derives a lower bound on the expected risk in the scalar-output case ($K=1$) by rewriting the trace term and applying standard eigenvalue inequalities, yielding
    \begin{align}\label{eq:lower_bound}
        \calL(\hvartheta_S) &\geq\calL'(S) + \frac{1}{N+1}\frac{\alpha\,\xi^2_{\min}\lambda_{\min}(C_f(\hvartheta_S))}{\lambda_{\min}(\hH_S(\hvartheta_S))}.
    \end{align}
    Here, $\xi^2_{\min}$ is the smallest non-zero value of $\lVert J_f((x,y);\vartheta)\rVert^2$, and the theorem assumes that data can be found for which $\lVert J_f((x,y);\vartheta)\rVert^2>0$ with probability $\alpha$ over the problem distribution.
    Further, $\lambda_{\min}$ denotes the smallest non-zero eigenvalue of a matrix.
    We note that the lower bound scales inversely with the smallest non-zero eigenvalue of the Hessian matrix of the empirical risk at local minima, $\lambda_{\min}(\hH_S(\hvartheta_S))$.
    The authors of Ref.~\cite{singh2022phenomenology} then carefully investigate the smallest non-zero eigenvalues of the matrices arising in classical neural networks.
    At this point, we depart from their formalism, and rather specialize to usual gradient-based PQCs.

\section{Double descent behavior in PQCs}

    \begin{figure*}[t]
        \centering
        \includegraphics[width=0.97\textwidth]{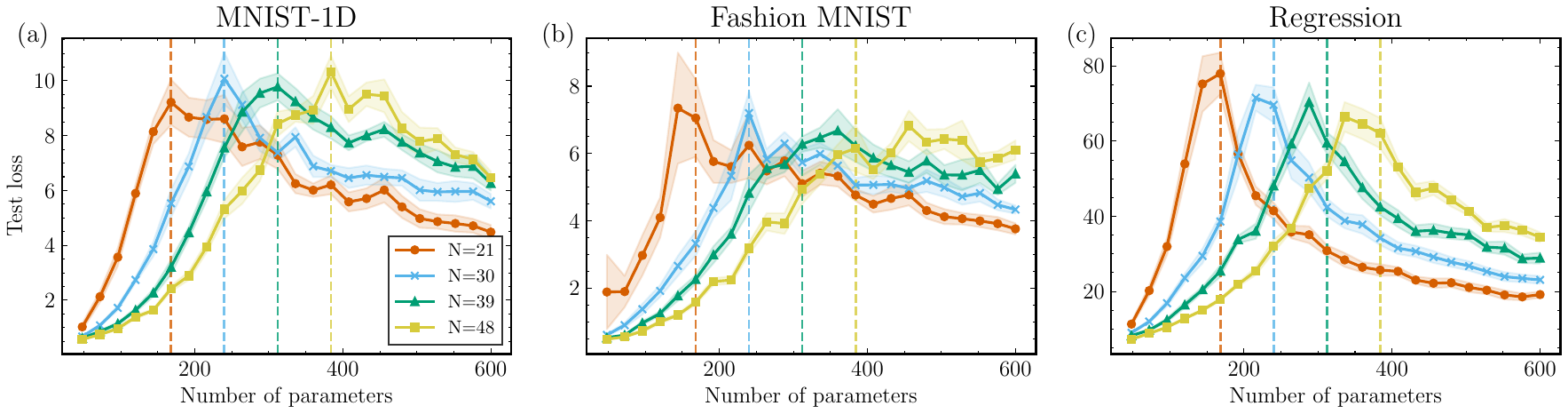} 
        \caption{
             Test loss as a function of the number of parameters $p$ for (a) MNIST-1D and (b) Fashion MNIST classification, and (c) a multidimensional synthetic regression task, for different training set sizes $N$. Vertical dashed lines indicate the predicted interpolation thresholds $p=NK$. In all three tasks, the test loss peaks close to interpolation and decreases again in the overparameterized regime, displaying a characteristic double descent profile.
            The shaded areas correspond to the standard deviation for ten independent experiment repetitions, each using independently sampled
            training data.
        }
        \label{fig:dd_main}
    \end{figure*}

    Our goal is to characterize the behavior of the lower bound on the expected risk in Eq.~\eqref{eq:lower_bound} as we vary the number of trainable parameters in our model $p$, as well as the number of training data $N$.
    In particular, we consider the asymptotic limit where $p,N\to\infty$ with fixed ratio $p/N\to\gamma$.
    The analysis relies on a series of formal assumptions, listed in Appendix~\ref{a:dd_behavior_in_PQCs} together with all further technical details.
    These assumptions are meant to capture the behavior of QML models in practice, when undergoing gradient-based training, and they mostly concern statements about the spectral properties of matrices related to the Hessian and the covariance matrices introduced above.
    We expect our assumptions to be largely ansatz-agnostic and to hold for a broad class of trainable PQCs, including data re-uploading architectures.
    We further test a selection of these assumptions numerically in Appendix~\ref{a:evidence}.
    
    We can now state our main result, which holds in the simpler case of learning tasks with a single output dimension $K=1$.
    
    \begin{theorem}[Double descent peak at interpolation -- informal]
        With the definitions above, and under reasonable assumptions, for Gaussian-distributed input data $x\sim\calN(0,\bbI_d)$, outputs $y\in\bbR$, and considering the mean-squared error loss, the lower bound on $\calL(\hvartheta_S)$ attains a maximum at $p=N$ with high probability in the limit $p,N\to\infty$ with fixed ratio $p/N\to\gamma$.
    \end{theorem}
    The formal version of this result is Theorem~\ref{thm:peakatinterpolation}, which can be found in Appendix~\ref{a:dd_behavior_in_PQCs}.
    The proof can be found in Appendix~\ref{a:proofofmaintheorem}.
    The main idea for the proof is to carefully control the spectral properties of several random matrices, which we achieve by combining our assumptions with seminal results in random matrix theory.
    Ultimately, we argue that the smallest non-zero eigenvalue of the dominant contribution to the Hessian matrix follows the behavior predicted by the so-called Mar\v{c}enko-Pastur law in the stated limit.
    Our analysis is inspired by the framework of Ref.~\cite{singh2022phenomenology}, with suitable modifications to the underlying assumptions.

    The assumption of i.i.d.\  Gaussian inputs is standard in high-dimensional random matrix analyses and has been extensively studied in classical literature~\cite{seddik2020random, couillet2022random}.
    The earlier work on double descent in QML~\cite{kempkes2025double} already argues that such assumptions often extend beyond strictly Gaussian data.
    Notably, the datasets considered in our numerical experiments in Section~\ref{s:empiricalevidence} are non-Gaussian, suggesting that the predicted behavior may extend beyond the assumptions of the theorem.
    
    Theorem~\ref{thm:peakatinterpolation} is restricted to the scalar-output case $K=1$.
    Nevertheless, the interpolation threshold in the proof arises from the spectral properties of a matrix of size $p\times NK$ in the general case, suggesting that the relevant transition should more generally occur when the number of parameters matches the number of scalar training constraints $NK$.
    We therefore conjecture that, under the assumptions of Theorem~\ref{thm:peakatinterpolation}, the lower bound on the expected risk exhibits a double descent peak at $p=NK$ also for $K>1$.
    Our numerical experiments provide evidence supporting this conjecture.

    Finally, in Appendix~\ref{app:upper} we provide further results concerning \emph{upper bounds} to the expected risk.
    Specifically, in Theorem~\ref{thm:upper_bound_K1} we give the upper bound
    \begin{align}
        \calL(\hvartheta_S) &\leq\calL'(S) + \frac{1}{N+1}\frac{\alpha\,\xi^2_{\max}\,R\,\lambda_{\max}(C_f(\hvartheta_S))}{\lambda_{\min}(\hH_S(\hvartheta_S))}.
    \end{align}
    Similarly as in Eq.~\eqref{eq:lower_bound}, this bound holds up to second-order corrections $\calO(N^{-2})$, and is dominated by spectral properties of the Hessian matrix $\hH_S(\hvartheta_S)$ and the covariance matrix $C_f(\hvartheta_S)$.
    Here $\xi_{\max}^2$ is defined analogously as before, and $R$ is the rank of the Hessian matrix, i.e., the number of non-zero eigenvalues.
    Under a slightly different set of assumptions, we show in Theorem~\ref{thm:upper_bound_peak} that this upper bound also peaks at interpolation $p=N$, with high probability in the same limit as above.

\section{Empirical evidence of double descent in PQCs}\label{s:empiricalevidence}

    The theoretical results presented in the previous section hold in the asymptotic limit and rely on assumptions about the local geometry of the loss landscape at local minima as well as the data being Gaussian.
    In this section, we provide numerical evidence that gradient-based data re-uploading PQCs display the predicted double descent behavior. 
    In particular, we show that, when training converges, the test loss, used here as a proxy for the expected risk, exhibits a pronounced peak near the interpolation threshold $p = NK$, where $p$ is the number of trainable parameters, $N$ is the training set size, and $K$ is the output dimension.
    
    We consider re-uploading PQCs where each input $x$ is encoded through angle embeddings on $n=8$ qubits, followed by trainable single-qubit rotations and nearest-neighbor $CZ$ gates.
    The model output is obtained by measuring Pauli-$X$ expectation values on $K$ qubits, giving $f_\theta(x)\in\mathbb{R}^K$.
    Before evaluating the loss, we rescale these expectation-value outputs by a fixed factor $c=150$.
    Empirically, we find that such a rescaling substantially improves optimization and makes convergence to local minima significantly more reliable.
    
    We perform numerical experiments on three datasets, shown in Fig.~\ref{fig:dd_main}.
    The first is the MNIST-1D dataset~\cite{Greydanus2020ScalingDD}, used as a multiclass classification task with $K=8$ classes.
    The second is the Fashion MNIST dataset~\cite{xiao2017fashion}, also restricted to $K=8$ classes.
    Finally, the third dataset is a synthetic multidimensional regression task with input dimension $d=8$ and output dimension $K=8$, generated from a noisy linear target function.
    For each dataset, we vary the number of training samples $N \in \{21,30,39,48\}$ and vary the circuit depth, thereby scanning across the underparameterized and overparameterized regimes.
    Further implementation details are provided in Appendix~\ref{app:numerics}.
    The vertical dashed lines in Fig.~\ref{fig:dd_main} mark the corresponding interpolation thresholds $p=NK$.
    
    These results provide empirical evidence that double descent appears in trainable PQCs optimized by gradient descent.
    Importantly, the second descent occurs in regimes where the number of trainable parameters exceeds the number of training samples, showing that increasing the size of the PQC further does not necessarily degrade performance on unseen data, contrary to the traditional expectation that larger models generalize worse.
    Instead, once the model passes the interpolation threshold and training succeeds, additional parameters can improve generalization.
    Note, however, that this does not imply that overparameterized models outperform underparameterized ones.
    Rather, it supports the main conclusion of our analysis: overparameterization can play a constructive role in gradient-based quantum machine learning.  
    We provide further details in Appendix~\ref{a:evidence} and some perspectives in
    Appendix~\ref{app:universality}.
    In particular, we show numerical results systematically testing the formal assumptions we introduce for Theorem~\ref{thm:peakatinterpolation}.

\section{Discussion}\label{s:discussion}

    In this work, we investigate the scalability of gradient-based \emph{parameterized quantum circuits} (PQCs) by analyzing how their expected risk evolves across different parameterization regimes.
    By combining an add-one-in perturbation analysis with spectral arguments from random matrix theory, and building on an existing lower bound on the expected risk, we show that the bound attains a maximum at the interpolation threshold.
    We further derive a corresponding upper bound and show that it exhibits the same qualitative dependence on model size.
    Our analysis therefore establishes that, under mild assumptions, the risk bounds of a broad class of PQCs exhibit double descent behavior.
    We confirm these findings through comprehensive numerical experiments with data re-uploading PQCs on both classification and regression tasks.
    
    Our work offers a complementary perspective to the growing literature on scalable QML under the slogan \enquote{train classical, deploy quantum}~\cite{huang2026spectral,recioarmengol2026trainclassical,kurkin2025universality,banks2026quditextension} and extends previous findings of double descent in quantum kernel methods~\cite{kempkes2025double}.
    The presented results indicate that overparameterized regimes may be compatible with favorable learning behavior~\cite{belkin2019reconciling,dascoli202double,dar2021farewell}, providing grounds for cautious optimism about the prospects of deep PQCs.

    The reason why we advocate for \emph{cautious} optimism is twofold.
    First, our analysis is restricted to trainable PQCs.
    Ref.~\cite{du2023problem} already pointed out that, in concentration-dominated regimes where trainability breaks down, larger models may not exhibit favorable generalization behavior.
    Our results should therefore be understood as a statement about what may be possible if trainability can be maintained as the number of parameters increases.
    We regard trainability as a prerequisite for practically relevant quantum machine learning and therefore focus on the generalization properties of models operating in such a regime.
    In this sense, we do not solve well-known trainability problems of QML~\cite{BarrenPlateaus,AnschuetzTraps,anschuetz2025unified, thanasilp2024exponentialconcentrationquantumkernel}.
    At the same time, we expect the assumptions underlying our analysis to hold for a broad class of trainable PQCs.
    Consequently, the qualitative phenomenon identified here may remain relevant for future trainable quantum models of practical relevance.

    The second reason for caution is that our results do not imply that overparameterized PQCs outperform underparameterized ones.
    Rather, they show that increasing the number of parameters beyond the interpolation threshold need not be detrimental to generalization, contrary to the traditional expectation from statistical learning theory that increasing model size ultimately leads to poorer generalization~\cite{vapnik2015uniform,Vapnik1999SLT,valiant1972learnable,ShalevShwartz2014understanding}.
    Determining when overparameterization leads to overall improvements in expected risk, and under which conditions such improvements can be realized in practice, remain important research directions.

    Another noteworthy observation from our numerical experiments is that rescaling the PQC outputs by a large constant improves convergence to local minima.
    One possible explanation is that expectation values are typically concentrated around zero, so that rescaling increases the density of low-training-loss solutions in parameter space by shifting the target labels away from the boundaries of the output interval.
    Understanding this effect theoretically is likely of independent interest and is left for future work.

    Overall, this work provides a framework for understanding how the generalization performance of parameterized quantum circuits evolves under scaling. We expect these results to inform future studies of learning, generalization, and overparameterization in increasingly large quantum models, and ultimately provide cautious optimism for deep parameterized quantum circuits.

\paragraph*{Code and data availability.}
The code and data used in this study are available in GitHub~\cite{github_repository}.

\paragraph*{Acknowledgments.} 
The authors would like to thank Evan Peters and Pablo Rodriguez-Grasa for insightful comments on an earlier version of this manuscript.
E.G.-F., C.B.-P.\, A.W.\ and J.E.\  acknowledge the BMFTR (MUNIQC-Atoms, HYBRID, QuSol, Hybrid++), the DFG (CRC 183 and SPP 2514), 
Bifold, 
the Quantum Flagship (Millenion, PasQuanS2), the Munich Quantum Valley, Berlin Quantum and the European Research Council (DebuQC) for financial support. E.G.-F. further acknowledges support from a 2023 Google PhD Fellowship.
V.D.\ acknowledges support from the European Union’s Horizon Europe program through the ERC CoG BeMAIQuantum (Grant No. 101124342). V.D. and E.vN acknowledge the Dutch National Growth Fund (NGF) as part of the Quantum Delta NL programme.
GPT‑5.6 was used to assist with code development and mathematical consistency checks during the preparation of this manuscript.
All scientific ideas, analyses, and conclusions are those of the authors, and all AI-assisted code and text were carefully reviewed, manually checked, and validated prior to inclusion.

\paragraph{Disclaimer.}
The results, opinions, and conclusions expressed in this publication are not necessarily those of Volkswagen Aktiengesellschaft or of Porsche Digital GmbH or those of the European Union or the European Research Council Executive Agency.
None of the parties involved can be held responsible for them.

\bibliography{references}

\vspace*{0.5cm}

\onecolumngrid
\appendix

\newpage

\setcounter{figure}{0}
\setcounter{theorem}{0}
\counterwithin{theorem}{section}
\counterwithin{figure}{section}

\setcounter{secnumdepth}{2}

\begin{center}
\large{Supplementary Material for \\``Cautious optimism for deep parameterized quantum circuits''}
\end{center}

\section{Framework}\label{a:framework}

    Throughout this work, we consider supervised learning tasks on a data space $\calZ \coloneqq \calX \times \calY$, where $\calX \subseteq \bbR^d$ denotes the input domain and $\calY \subseteq \bbR^K$ the output domain.
    Samples $z = (x, y)$ are assumed to be drawn i.i.d.\ from an unknown but fixed distribution $\calD$ on $\calZ$.
    Given a training set $S=\{z_i\}_{i=1}^N$, the goal is to learn a hypothesis $f_\vartheta\in\calF$, where $\calF=\{f_\vartheta:\calX\to\calY\mid \vartheta\in\Theta\}$ denotes the hypothesis class and $\Theta\subseteq\bbR^p$ the parameter space.
    The hypothesis class is realized by \emph{parameterized quantum circuits}  (PQCs).
    For a classical input $x$, the model prepares a quantum state
    \begin{align}
        \rho(x,\vartheta) =  U(x,\vartheta)\rho_0U^\dagger(x,\vartheta),
    \end{align}      
    where $\rho_0$ is the initial state and $U(x,\vartheta)$ is the unitary implemented by the PQC.
    A common architecture consists of alternating data-dependent and parameter-dependent layers,
    \begin{align}
        U(x,\vartheta) = \prod_{l=1}^{L}V_l(\vartheta)U_l(x),
    \label{eq:pqc_ansatz}
    \end{align}    
    commonly referred to as data re-uploading~\cite{perez2020data}.
    We consider vector-valued models $f_\vartheta:\calX\rightarrow\calY\subseteq\bbR^K$ with outputs
    \begin{align}
        f_\vartheta(x) = \bigl([f_\vartheta(x)]_1,\ldots [f_\vartheta(x)]_K\bigr),
    \end{align}    
    where each component is obtained as the expectation value
    \begin{align}
        [f_\vartheta(x)]_k = \Tr\left\{\rho(x,\vartheta)O_k\right\}, \qquad k=1,\ldots,K,
    \label{eq:PQCoutput}
    \end{align}
    of a Hermitian observable.
    For classification tasks with $K>1$, the model output $f_\vartheta(x)$ is interpreted as a vector of class scores, while the predicted label is typically given by $\argmax_k [f_\vartheta(x)]_k$.
    For $K=1$, a setting commonly used in binary classification, the predicted label can instead be obtained by thresholding the output, e.g., via $\mathrm{sign}(f_\vartheta(x))$.

    Let a learning task be specified by a data distribution $\calD(\calZ)$, which we assume to be fixed and unknown.
    We denote by $S=\{(x_i,y_i)\}_{i=1}^N$ a training set of size $N$ composed of i.i.d.\ samples $z_i\coloneqq(x_i,y_i)$ from the problem distribution $S\sim\calD^N$.
    Next, let $\ell\colon\calZ\times\Theta\to\bbR_{\geq0}$ be a twice-differentiable loss function, which given $z\in\calZ$ and a parameter specification $\vartheta\in\Theta$ acts as a measure of distance between the prediction $f_\vartheta(x)$ and the correct label $y$.
    From the loss function $\ell$ and the problem distribution $\calD$ we define the expected risk functional $\calL\colon\Theta\to\bbR_{\geq0}$ as
    \begin{align}
        \calL(\vartheta) &\coloneqq\bbE_{z\sim D}\left[\ell(z,\vartheta)\right].
    \end{align}
    The goal in supervised learning is to find a set of parameters that achieve a small expected risk, though this task cannot be tackled head-on due to the assumption that $\calD$ is unknown.
    Rather, the common approach is to use a proxy quantity relative to a given training set $S$, which we denote \emph{empirical risk} $\hcalL_S\colon\Theta\to\bbR_{\geq0}$, defined as
    \begin{align}
        \hcalL_S(\vartheta) &\coloneqq \frac{1}{N}\sum_{i=1}^N\ell(z_i,\vartheta).
    \end{align}
    This way, the empirical risk is an unbiased estimator for the expected risk.
    In the following, we abuse language in that we identify a function $f_\vartheta$ (also called \emph{hypothesis}) with its parameters $\vartheta$.

    \begin{definition}[Local minimum]\label{def:localminimum}
        Given the model $\calF$, a training set $S$, and the empirical risk functional $\hcalL_S$, we call a parameter vector \emph{a local minimum}, denoted by $\hvartheta_S$, if it fulfills
        \begin{align}
            \nabla_\vartheta\hcalL_S(\hvartheta_S) &= 0,
            \\
            \hH_S(\hvartheta_S) &\succeq 0.
        \end{align}
    \end{definition}

    To avoid confusion with other differential operators that appear below, $\nabla_\vartheta$ represents the gradient operator 
    \begin{align}
        \nabla_\vartheta &\coloneqq\begin{pmatrix}
            \pder{\vartheta_1} & \cdots & \pder{\vartheta_p}
        \end{pmatrix}^\intercal
    \end{align}
    with respect to the parameters. 
    Specifically, for any $\vartheta\in\Theta$, the gradient of the loss with respect to the parameters is a $p$-dimensional column vector $\nabla_\vartheta\hcalL_S(\vartheta)\in\bbR^p$.

    \begin{remark}{(Local minima as output){\bf .}}
        From here on, we make the explicit assumption that standard gradient-based training techniques produce local minima as output.
    \end{remark}
    We denote by $\hH_S(\vartheta)$ the Hessian matrix of the empirical risk
    \begin{align}
        \hH_S(\vartheta) &= \nabla_\vartheta\nabla_\vartheta^\intercal\hcalL_S(\vartheta) = \left(\pder[^2\hcalL_S(\vartheta)]{\vartheta_i\partial\vartheta_j}\right)_{i,j=1}^p.
    \end{align}
    
    \begin{remark}[The inverse of the Hessian.]
        Our results below are stated in terms of the inverse of the Hessian matrix, which does not exist unless the matrix is full rank.
        Throughout this work, we therefore interpret $[\hH_S(\vartheta)]^{-1}$ as the Moore-Penrose pseudoinverse of $\hH_S(\vartheta)$, obtained by inverting only the non-zero eigenvalues.
        This is particularly relevant in overparameterized models, where redundant parameters give rise to zero eigenvalues.
    \end{remark}

\section{Expected risk at local minima}\label{a:expected_risk_local_minima}

    Since double descent is fundamentally a phenomenon of the expected risk, we first derive an explicit expression for the expected risk attained attained at local minima.  
    To this end, we analyze the effect of adding a single sample to the training set.
    Following Ref.~\cite{singh2022phenomenology}, we consider an add-one-in perturbation of the empirical risk.
    Under this perturbation, the empirical risk changes only slightly, and the corresponding local minimum is therefore expected to remain close to that of the original training set.
    This motivates a local expansion around the original local minimum.

    \begin{definition}[Add-one-in perturbation]\label{def:perturbation}
        Given a training set $S$ and a datum $z'\in\calZ$ not already-contained in $S$, we refer to $S'=S\cup\{z'\}$ as the \emph{add-one-in perturbation} of $S$.
    \end{definition}
    
    We do not include the $z'$ dependence of $S'$ explicitly because below we only make statements that hold in expectation over $z'\sim\calD$.
    The distinction we keep in mind regarding notation is the size of the sets: for $\lvert S\rvert=N$ we have $\lvert S'\rvert=N+1$.
    Note that if we call $\calD_S$ the uniform distribution over $S$, and similarly for $S'$, then $\calD_{S'}=(1-\varepsilon)\calD_S+\varepsilon\delta_{z'}$.
    Here $\delta_{z'}$ is a delta distribution centered at $z'$, and 
    \begin{equation}\varepsilon=\frac{1}{N+1}.
    \end{equation}
    In the regime of large $N$, we have that $\varepsilon$ becomes small, which justifies calling $\calD_{S'}$ a \enquote{perturbation} of $\calD_{S}$.
    Again, we abuse language and we identify the sets $S$ and $S'$ with their uniform distributions.

    \begin{lemma}[{Influence function for add-one-in perturbation, adapted from Ref.~\cite[Prop.~2]{singh2022phenomenology}}]\label{l:IF}
        Let $\calF$ be a hypothesis family, $S$ a training set of size $N$.
        Let $S'$ be the add-one-in perturbation of $S$, and let $\hvartheta_{S}, \hvartheta_{S'}$ be local minima with respect to $S$ and $S'$ respectively.
        Then, the first-order contribution from adding samples to the training set is as 
        \begin{align}
            \frac{\ell(z',\hvartheta_{S'})-\ell(z',\hvartheta_S)}{\varepsilon} &= -\nabla_\vartheta\ell(z',\hvartheta_S)^\intercal\left[\hH_S(\hvartheta_S)\right]^{-1}\nabla_\vartheta\ell(z',\hvartheta_S) + \calO\left(\frac{1}{N}\right).
        \end{align}
    \end{lemma}

    \begin{definition}[Add-one-in loss]\label{def:add-loss}
        Suppose we are given the model $\calF$, a training set $S$ and its corresponding local minimum $\hvartheta_S$.
        Further, for any add-one-in perturbation $S'=S\cup\{z'\}$, let $\hvartheta_{S'}$ be its local minimum.
        Then, we denote the \emph{add-one-in loss} $\calL'(S)$ as
        \begin{align}
            \calL'(S) &\coloneqq \bbE_{z'\sim\calD} \left[\ell(z',\hvartheta_{S'})\right].
        \end{align}
    \end{definition}

    \begin{definition}[Uncentered covariance of loss gradients]
    For any parameter vector $\vartheta\in\mathbb{R}^p$, we define the \emph{uncentered covariance matrix of the loss gradients} as
    \begin{align}
        C(\vartheta)
        \coloneqq
        \bbE_{z\sim\calD}\left[
            \nabla_\vartheta\ell(z,\vartheta)
            \nabla_\vartheta\ell(z,\vartheta)^\intercal
        \right].
    \end{align}
    \end{definition}

    \begin{lemma}[Error decomposition, Theorem 3 in Ref.~\cite{singh2022phenomenology}]\label{l:error_decomp}
        Let $\calF$ be the hypothesis family of a learning model, let $S$ be a training set of size $N$, let $\ell$ be a twice-differentiable loss function, let $\hvartheta_S$ be a local minimum, let $\calL'(S)$ be the add-one-in loss, and let $\hH_S(\hvartheta_S)$ be the Hessian matrix of the empirical loss at the local minimum.
        Then, the expected risk $\calL(\hvartheta_S)$ fulfills
        \begin{align}
            \calL(\hvartheta_S) &= \calL'(S) + \frac{1}{N+1}\Tr\left\{\left[\hH_S(\hvartheta_S)\right]^{-1}C(\hvartheta_S)\right\}+\calO\left(\frac{1}{N^2}\right).
        \label{eq:error_decomp}
        \end{align}
        We refer to the summand with the trace as the \emph{complexity term}.
    \end{lemma}
    \begin{proof}
        We prove this statement directly, starting from Lemma~\ref{l:IF} and taking an expectation over $\calD$.
        We start by taking the equation in the Lemma and substituting \begin{align}
        \varepsilon=\frac{1}{N+1},
        \end{align}
        in order to get
        \begin{align}
            \frac{\ell(z',\hvartheta_{S'})-\ell(z',\hvartheta_S)}{\varepsilon} &= -\nabla_\vartheta\ell(z',\hvartheta_S)^\intercal\left[\hH_S(\hvartheta_S)\right]^{-1}\nabla_\vartheta\ell(z',\hvartheta_S) + \calO\left(\frac{1}{N}\right) \\
            \ell(z',\hvartheta_{S'})-\ell(z',\hvartheta_S) &= -\varepsilon\nabla_\vartheta\ell(z',\hvartheta_S)^\intercal\left[\hH_S(\hvartheta_S)\right]^{-1}\nabla_\vartheta\ell(z',\hvartheta_S) + \calO\left(\frac{1}{N^2}\right) \\
            &= -\frac{1}{N+1}\nabla_\vartheta\ell(z',\hvartheta_S)^\intercal\left[\hH_S(\hvartheta_S)\right]^{-1}\nabla_\vartheta\ell(z',\hvartheta_S) + \calO\left(\frac{1}{N^2}\right).
        \end{align}
        Next we perform a common step to turn a vector-matrix-vector product into the trace of a matrix-matrix product as
        \begin{align}
            \ell(z',\hvartheta_{S'})-\ell(z',\hvartheta_S) &= -\frac{1}{N+1}\nabla_\vartheta\ell(z',\hvartheta_S)^\intercal\left[\hH_S(\hvartheta_S)\right]^{-1}\nabla_\vartheta\ell(z',\hvartheta_S) + \calO\left(\frac{1}{N^2}\right) \\
            &= -\frac{1}{N+1}\Tr\left\{\nabla_\vartheta\ell(z',\hvartheta_S)^\intercal\left[\hH_S(\hvartheta_S)\right]^{-1}\nabla_\vartheta\ell(z',\hvartheta_S)\right\} + \calO\left(\frac{1}{N^2}\right) \\
            &= -\frac{1}{N+1}\Tr\left\{\left[\hH_S(\hvartheta_S)\right]^{-1}\nabla_\vartheta\ell(z',\hvartheta_S)\nabla_\vartheta\ell(z',\hvartheta_S)^\intercal\right\} + \calO\left(\frac{1}{N^2}\right).
        \end{align}
        We now take the expectation over $z'\sim\calD$ on both sides and simplify the terms that do not depend on $z'$, in order to get
        \begin{align}
            \bbE_{z'\sim\calD} \left[\ell(z',\hvartheta_{S'})-\ell(z',\hvartheta_S)\right] &= \bbE_{z'\sim\calD}\left[-\frac{1}{N+1}\Tr\left\{\left[\hH_S(\hvartheta_S)\right]^{-1}\nabla_\vartheta\ell(z',\hvartheta_S)\nabla_\vartheta\ell(z',\hvartheta_S)^\intercal\right\} + \calO\left(\frac{1}{N^2}\right)\right] ,
            \\
            \bbE_{z'\sim\calD} \left[\ell(z',\hvartheta_{S'})\right] - \bbE_{z'\sim\calD} \left[\ell(z',\hvartheta_S)\right] &= -\frac{1}{N+1}\Tr\left\{\left[\hH_S(\hvartheta_S)\right]^{-1}\bbE_{z'\sim\calD}\left[\nabla_\vartheta\ell(z',\hvartheta_S)\nabla_\vartheta\ell(z',\hvartheta_S)^\intercal\right]\right\} + \calO\left(\frac{1}{N^2}\right).
        \end{align}
        Notice that the two expectations on the left-hand side are precisely the add-one-in loss and the expected risk, and the expectation value on the right-hand side corresponds to the uncentered covariance of the gradients.
        From here we need only shift terms around to complete the proof,
        \begin{align}
            \calL'(S) - \calL(\hvartheta_S) &= -\frac{1}{N+1}\Tr\left\{\left[\hH_S(\hvartheta_S)\right]^{-1}C(\hvartheta_S)\right\} + \calO\left(\frac{1}{N^2}\right) \\
            \calL(\hvartheta_S) &= \calL'(S) + \frac{1}{N+1}\Tr\left\{\left[\hH_S(\hvartheta_S)\right]^{-1}C(\hvartheta_S)\right\} + \calO\left(\frac{1}{N^2}\right).
        \end{align}
    \end{proof}

    \begin{remark}
     Recall that $[\hH_S(\hvartheta_S)]^{-1}$ denotes the Moore-Penrose pseudoinverse.
     Consequently, directions corresponding to zero Hessian eigenvalues do not contribute to the complexity term $\Tr\left\{[\hH_S(\hvartheta_S)]^{-1}C(\hvartheta_S)\right\}$.
     Therefore, Lemma~\ref{l:error_decomp} does not contradict the common intuition that flat minima may generalize better.
     Rather, it shows that the expected risk is increased by directions that simultaneously exhibit small non-zero curvature and large gradient variance as quantified by $C(\hvartheta_S)$.
     Intuitively, such directions are only weakly constrained by the training set, despite corresponding to directions along which typical samples from the underlying distribution exhibit substantial gradients.
    \end{remark}    

    We next reproduce the lower bound for this error decomposition, following the formalism in Ref.~\cite{singh2022phenomenology}.
    To this end, we restrict our attention to the scalar-output case $K=1$ and introduce a mild regularity assumption on the loss gradients.
    We start by introducing two auxiliary lemmas that will be used in the proof.
    
    \begin{lemma}\label{l:eval_trace_ineq}
        Let $A,B\in\bbR^{m\times m}$ be symmetric matrices, with $A, B\geq0$ further being positive semi-definite.
        Then
        \begin{align}
            \lambda_{\min}(A)\Tr\{B\} &\leq \Tr\{AB\} \leq \lambda_{\max}(A)\Tr\{B\}.
        \end{align}
        Here, $\lambda_{\min}(A)$ corresponds to the smallest eigenvalue of $A$, and similarly for $\lambda_{\max}$.
    \end{lemma}
    
    The proof of Lemma~\ref{l:eval_trace_ineq} can be found in Ref.~\cite{fang1994inequalities}.

    \begin{lemma}\label{l:accompany_proof_bounds}
        Under Assumption~\ref{A1} and in the case $K=1$, the complexity term from Lemma~\ref{l:error_decomp} fulfills
        \begin{align}
            \Tr\left\{\left[\hH_S(\hvartheta_S)\right]^{-1}C(\hvartheta_S)\right\} &= \alpha\bbE_{\substack{z\sim\calD\\\xi_z^2>0}}\left[\xi_z^2 \Tr\left\{\left[\hH_S(\hvartheta_S)\right]^{-1}\nabla_\vartheta f_{\hvartheta_S}(z) \nabla_\vartheta f_{\hvartheta_S}(z)^\intercal\right\}\right].
        \end{align}
    \end{lemma}
    \begin{proof}
        We prove this statement directly.
        We first use the chain rule to relate $\nabla_\vartheta\ell(z,\vartheta)$ and $\nabla_\vartheta f_\vartheta(z)$, and then we separate the expectation value into the two disjoint conditions $\xi_z^2>0$ and $\xi_z^2=0$,
        \begin{align}
            \nabla_\vartheta\ell(z,\vartheta) &\coloneqq \left(\pder[\ell(z,\hvartheta_s)]{\vartheta_m}\right)_{m=1}^p 
            = \left(\pder[\ell(z,\vartheta)]{f_\vartheta(x)}\pder[f_\vartheta(z)]{\vartheta_m}\right)_{m=1}^p
            = \pder[\ell(z,\vartheta)]{f_\vartheta(x)}\nabla_\vartheta f_\vartheta(x).
        \end{align}
        In the case $K=1$, $\pder[\ell(z,\vartheta)]{f_\vartheta(x)}\in\bbR$ is a scalar and $\nabla_\vartheta f_\vartheta(x)\in\bbR^p$ is a $p$-dimensional column vector.
        We now substitute the chain rule into the definition of $C(\hvartheta_S)$, to obtain
        \begin{align}
            C(\hvartheta_S) &= \bbE_{z'\sim\calD}\left[\nabla_\vartheta\ell(z',\hvartheta_S)\nabla_\vartheta\ell(z',\hvartheta_S)^\intercal\right]
            = \bbE_{z'\sim\calD}\left[\left\lvert\pder[\ell(z',\hvartheta_S)]{f_\vartheta(x)}\right\rvert^2 \nabla_\vartheta f_{\hvartheta_S}(x')\nabla_\vartheta f_{\hvartheta_S}(x')^\intercal\right].
        \end{align}
        The squared term inside the expectation value is exactly $\lVert\nabla_f\ell(z',\hvartheta_S)\rVert^2$ in the case of $K=1$, which we defined as $\xi^2_{z'}$ in Assumption~\ref{A1},
        \begin{align}
            \bbE_{z'\sim\calD}\left[\left\lvert\pder[\ell(z',\hvartheta_S)]{f_\vartheta(x)}\right\rvert^2 \nabla_\vartheta f_{\hvartheta_S}(x')\nabla_\vartheta f_{\hvartheta_S}(x')^\intercal\right] &= \bbE_{z'\sim\calD}\left[\xi_{z'}^2 \nabla_\vartheta f_{\hvartheta_S}(x')\nabla_\vartheta f_{\hvartheta_S}(x')^\intercal\right].
        \end{align}
        Now, we split the expectation value over two disjoint conditions $\bbE_{z'\sim\calD}[\,\cdot\,] = \bbP(\xi_{z'}^2>0)\bbE_{\substack{z'\sim\calD\\\xi_{z'}^2>0}}[\,\cdot\,] + \bbP(\xi_{z'}^2=0)\bbE_{\substack{z'\sim\calD\\\xi_{z'}^2=0}}[\,\cdot\,]$.
        Assumption~\ref{A1} can be restated as: there exists an $\alpha>0$ such that $\bbP_{z\sim\calD}(\xi^2_z>0)=\alpha$.
        With this, we obtain
        \begin{align}
            \bbE_{z'\sim\calD}\left[\xi_{z'}^2 \nabla_\vartheta f_{\hvartheta_S}(x')\nabla_\vartheta f_{\hvartheta_S}(x')^\intercal\right] &= \bbP(\xi_{z'}^2>0)\bbE_{\substack{z'\sim\calD\\\xi_{z'}^2>0}}\left[\xi_{z'}^2 \nabla_\vartheta f_{\hvartheta_S}(x')\nabla_\vartheta f_{\hvartheta_S}(x')^\intercal\right] \\
            &\hphantom{\bbP}+ \bbP(\xi_{z'}^2=0)\bbE_{\substack{z'\sim\calD\\\xi_{z'}^2=0}}\left[\xi_{z'}^2 \nabla_\vartheta f_{\hvartheta_S}(x')\nabla_\vartheta f_{\hvartheta_S}(x')^\intercal\right] \\
            &= \alpha \bbE_{\substack{z'\sim\calD\\\xi_{z'}^2>0}}\left[\xi_{z'}^2 \nabla_\vartheta f_{\hvartheta_S}(x')\nabla_\vartheta f_{\hvartheta_S}(x')^\intercal\right].
        \end{align}
        The second summand vanishes due to the $\xi_{z'}^2=0$ condition.
        Directly substituting these in the complexity term finishes the proof,
        \begin{align}
            \Tr\left\{\left[\hH_S(\hvartheta_S)\right]^{-1}C(\hvartheta_S)\right\} &= \Tr\left\{\left[\hH_S(\hvartheta_S)\right]^{-1}\alpha \bbE_{\substack{z'\sim\calD\\\xi_{z'}^2>0}}\left[\xi_{z'}^2 \nabla_\vartheta f_{\hvartheta_S}(x')\nabla_\vartheta f_{\hvartheta_S}(x')^\intercal\right]\right\} \\
            &= \alpha \bbE_{\substack{z'\sim\calD\\\xi_{z'}^2>0}}\left[\xi_{z'}^2\Tr\left\{\left[\hH_S(\hvartheta_S)\right]^{-1}\nabla_\vartheta f_{\hvartheta_S}(x')\nabla_\vartheta f_{\hvartheta_S}(x')^\intercal\right]\right\}.
        \end{align}
    \end{proof}
    
    \begin{assumption}[Regularity assumption]
        There exists a sample $z \sim \mathcal{D}$ with non-zero probability $\alpha$ such that $\xi^2_z \coloneqq \lVert \nabla_f \ell(z, \hvartheta_S)\rVert^2 > 0$, where $\nabla_f$ denotes the gradient with respect to the model output.
    \label{A1}
    \end{assumption}

    For the mean-squared error loss considered in this work, Assumption~\ref{A1} merely requires that predictions of the model with parameters $\hvartheta_S$ differ from the true targets on a subset of the data distribution with non-zero probability.

    \begin{definition}[Uncentered covariance of function gradients]
        For any parameter vector $\vartheta\in\mathbb{R}^p$, we define the \emph{uncentered covariance matrix of the function gradients} as
        \begin{align}
            C_f(\vartheta) \coloneqq \bbE_{\substack{z\sim\calD\\\xi_z^2>0}} \left[ \nabla_\vartheta f_{\vartheta}(x) \nabla_\vartheta f_{\vartheta}(x)^\intercal \right].
        \end{align}
    \end{definition}

    Combining the error decomposition from Lemma~\ref{l:error_decomp} with Assumption~\ref{A1} yields the following lower bound on the expected risk in terms of the spectrum of $C_f(\hvartheta_S)$ and the Hessian matrix of the empirical risk $\hH_S(\hvartheta_S)$ at local minima.

    \begin{theorem}[Lower bound on the expected risk, Theorem 4 Ref.~\cite{singh2022phenomenology}]\label{thm:lower_bound_K1}
        Under the setting of Lemma~\ref{l:error_decomp}, under Assumption~\ref{A1} and for $K=1$, the expected risk $\calL(\hvartheta_S)$ at the local minimum $\hvartheta_S$ fulfills the lower bound
    \begin{align}
            \calL(\hvartheta_S) &\geq\calL'(S) + \frac{1}{N+1}\frac{\alpha\,\xi^2_{\min}\lambda_{\min}(C_f(\hvartheta_S))}{\lambda_{\min}(\hH_S(\hvartheta_S))}.
        \end{align}
        Here, we define $\xi^2_{\min} = \min_{\substack{z\sim\calD \\ \xi_z^2>0}} \{\xi_z^2\}$, and $\lambda_{\min}(\cdot)$ the smallest non-zero eigenvalue of its matrix argument.
    \end{theorem}
    \begin{proof}
        The proof follows directly from Lemmas~\ref{l:error_decomp},~\ref{l:eval_trace_ineq}, and~\ref{l:accompany_proof_bounds}.

        We first provide a lower bound for the complexity term in Lemma~\ref{l:accompany_proof_bounds} using the definition of $\xi_{\min}^2$
        \begin{align}
            \Tr\left\{\left[\hH_S(\hvartheta_S)\right]^{-1}C(\hvartheta_S)\right\} &= \alpha\bbE_{\substack{z\sim\calD\\
             \xi_z^2>0}}\left[\xi_z^2 \Tr\left\{\left[\hH_S(\hvartheta_S)\right]^{-1}\nabla_\vartheta f_{\hvartheta_S}(z) \nabla_\vartheta f_{\hvartheta_S}(z)^\intercal\right\}\right] \\
            &\geq \alpha\,\xi_{\min}^2 \Tr\left\{\left[\hH_S(\hvartheta_S)\right]^{-1}\bbE_{\substack{z\sim\calD\\\xi_z^2>0}}\left[\nabla_\vartheta f_{\hvartheta_S}(z) \nabla_\vartheta f_{\hvartheta_S}(z)^\intercal\right\}\right\} \\
            &= \alpha\,\xi_{\min}^2 \Tr\left\{\left[\hH_S(\hvartheta_S)\right]^{-1} C_f(\hvartheta_S)\right\}.
        \end{align}
        Lemma~\ref{l:eval_trace_ineq} allows us to lower-bound this quantity in terms of the eigenvalues of $C_f(\hvartheta_S)$ as
        \begin{align}
            \Tr\left\{\left[\hH_S(\hvartheta_S)\right]^{-1} C_f(\hvartheta_S)\right\} & \geq \lambda_{\min}(C_f(\hvartheta_S))\Tr\left\{\left[\hH_S(\hvartheta_S)\right]^{-1}\right\}.
        \end{align}
        Now, under the assumption that $\hH_S(\hvartheta_S)$ is full-rank, we can immediately relate the trace of its inverse to the eigenvalues of the Hessian.
        Namely: the trace of a matrix is the sum of its eigenvalues, and for Hermitian matrices, the eigenvalues of the inverse are the inverse of the eigenvalues $\lambda_i([\hH_S(\hvartheta_S)]^{-1}) = (\lambda_i(\hH_S(\hvartheta_S)))^{-1}$,
        \begin{align}
            \Tr\left\{\left[\hH_S(\hvartheta_S)\right]^{-1}\right\} &= \sum_{i=1}^p \lambda_i\left(\left[\hH_S(\hvartheta_S)\right]^{-1}\right) \\
            &= \sum_{i=1}^R \frac{1}{\lambda_i\left(\hH_S(\hvartheta_S)\right)} \\
            &\geq \frac{1}{\lambda_{\min}(\hH_S(\hvartheta_S))}.
        \end{align}
        We call $R=\operatorname{rank}(\hH_S(\hvartheta_S))$, namely the number of non-zero eigenvalues of the Hessian.
        From the first to the second line, we replace a sum over $p$ eigenvalues (dimension of the matrix) with a sum over the $R$-many inverses of non-zero eigenvalues.
        In the last step, we lower-bounded the value of a sum of positive numbers by the value of the largest element in the sum.
        We can finally insert these bounds in the error decomposition (Eq.~\eqref{eq:error_decomp} from Lemma~\ref{l:error_decomp}), to get
        \begin{align}
            \calL(\hvartheta_S) &=\calL'(S) + \frac{1}{N+1}\Tr\left\{\left[\hH_S(\hvartheta_S)\right]^{-1}C(\hvartheta_S)\right\}+\calO\left(\frac{1}{N^2}\right) \\
            &\geq \calL'(S) + \frac{1}{N+1}\left(\alpha\,\xi_{\min}^2 \Tr\left\{\left[\hH_S(\hvartheta_S)\right]^{-1} C_f(\hvartheta_S)\right\}\right)+\calO\left(\frac{1}{N^2}\right) \\
            &\geq \calL'(S) + \frac{\alpha\,\xi_{\min}^2}{N+1}\left( \lambda_{\min}(C_f(\hvartheta_S))\Tr\left\{\left[\hH_S(\hvartheta_S)\right]^{-1}\right\}\right)+\calO\left(\frac{1}{N^2}\right) \\
            &\geq \calL'(S) + \frac{\alpha\,\xi_{\min}^2\lambda_{\min}(C_f(\hvartheta_S))}{N+1}\left( \frac{1}{\lambda_{\min}(\hH_S(\hvartheta_S))}\right)+\calO\left(\frac{1}{N^2}\right) \\
            &=\calL'(S) + \frac{1}{N+1}\frac{\alpha\,\xi^2_{\min}\lambda_{\min}(C_f(\hvartheta_S))}{\lambda_{\min}(\hH_S(\hvartheta_S))} + \calO\left(\frac{1}{N^2}\right).
        \end{align}
        This completes the proof.
    \end{proof}
    In particular, the lower bound is governed by the smallest non-zero eigenvalue of the Hessian matrix of the empirical risk at local minima, $\lambda_{\min}(\hH_S(\hvartheta_S))$, a quantity that will play a central role in our subsequent analysis.
    A complementary upper bound, obtained using analogous techniques, is presented in Theorem~\ref{thm:upper_bound_K1} in Appendix~\ref{app:upper}.

\section{Double descent behavior in PQCs}\label{a:dd_behavior_in_PQCs}

    We now establish that, as the number of trainable parameters $p$ is varied, the lower bound on the expected risk exhibits a double descent peak.
    While this does not directly imply that the expected risk itself exhibits double descent, numerical results presented later provide evidence for a corresponding phenomenon.
    Restricting to the scalar-output case $K=1$, we show that the peak occurs at $p=N$, corresponding to the interpolation threshold.
    Numerical results presented later indicate the same qualitative behavior for $K>1$.
    We therefore keep the general notation in terms of $K$ throughout, although the rigorous results derived below are restricted to $K=1$.
    Our analysis is inspired by the framework of Ref.~\cite{singh2022phenomenology}, with suitable modifications to the underlying assumptions.

    A central ingredient in the derivation is a decomposition of the Hessian of the empirical risk into the so-called outer-product Hessian $\hH_S^o(\vartheta)$ and the functional Hessian $\hH_S^f(\vartheta)$,
    \begin{align}\label{eq:Hessian_decomp}
        \hH_S(\vartheta) = \hH_S^o(\vartheta) + \hH_S^f(\vartheta),
    \end{align}
    where
    \begin{align}
        \hH_S^o(\vartheta) &= \frac{1}{N}\sum_{i=1}^{N} \nabla_\vartheta f_\vartheta(x_i) \left[ \nabla_f\nabla_f^\intercal\ell(z_i,\vartheta) \right] \nabla_\vartheta f_\vartheta(x_i)^\intercal,\\
        \hH_S^f(\vartheta) &=\frac{1}{N}\sum_{i=1}^{N} \nabla_f^\intercal \ell(z_i,\vartheta) \nabla_\vartheta\nabla_\vartheta^\intercal f_\vartheta(x_i).
    \end{align}

    To prove our main result, Theorem~\ref{thm:peakatinterpolation}, we require several assumptions on the quantities appearing in the Hessian decomposition and the ensuing spectral analysis.
    This section culminates with the statement of Theorem~\ref{thm:peakatinterpolation}, and its proof (including a derivation for the decomposition of the Hessian) is provided in Appendix~\ref{a:proofofmaintheorem} in a self-contained way.
    We first introduce two quantities that will play a central role in formulating these assumptions.

    \begin{definition}[Sample Jacobian of the function]
    For a training set $S=\{(x_i,y_i)\}_{i=1}^N$ and parameter vector
    $\vartheta\in\mathbb{R}^p$, the sample Jacobian of the function is defined as
    \begin{align}
        Z_S(\vartheta)\coloneqq\bigl[\nabla_\vartheta f_{\vartheta}(x_1),\ldots,\nabla_\vartheta f_{\vartheta}(x_N)\bigr] \in \bbR^{p \times NK}. 
    \end{align}
    \end{definition}
    
    \begin{definition}[Uncentered sample covariance of function gradients]
        For a training set $S=\{(x_i, y_i)\}_{i=1}^N$ and parameter vector
        $\vartheta\in\mathbb{R}^p$, the uncentered sample covariance of the function gradients is defined as
        \begin{align}
            \hC_S^f(\vartheta)\coloneqq\frac{1}{N}Z_S(\vartheta)Z_S(\vartheta)^\intercal \in \bbR^{p\times p}.
        \end{align}
    \end{definition}

    The following assumptions serve different purposes in the analysis.
    Assumption~\ref{A2} concerns the behavior of the functional Hessian near interpolation, while the subsequent assumptions concern properties of the outer-product Hessian and the covariance of function gradients.

    \begin{assumption}
        In the underparameterized regime $p < NK$, the largest eigenvalue        $\lambda_{\max}\bigl(\hH_S^f(\hvartheta_S)\bigr)$ of the functional Hessian at local minima decreases with the number of trainable parameters $p$ in a neighborhood of the interpolation threshold $p=NK$.
    \label{A2}
    \end{assumption}

    For the mean-squared error loss, the functional Hessian depends explicitly on the training residuals.
    Assumption~\ref{A2} is therefore consistent with the expectation that these residuals decrease as the number of trainable parameters increases.
    We investigate this assumption numerically later.
    
    The next assumption concerns the second term in the Hessian decomposition, namely the outer-product Hessian $\hH_S^o(\vartheta)$.
    For the scalar-output mean-squared error loss, this term reduces to the uncentered sample covariance of the function gradients, $\hC_S^f(\vartheta)$.
    This follows immediately from the fact that the first derivative of the loss is $\nabla_f\ell(z_i,\vartheta)=(f_\vartheta(x_i)-y_i)$ while its second derivative $\nabla_f(f_\vartheta(x_i)-y_i))=1$ .
    We can, therefore,  write    
    \begin{align}
        \hH_S^o(\vartheta) = \hC_S^f(\vartheta) = \frac{1}{N} Z_S(\vartheta) Z_S(\vartheta)^\intercal.
    \end{align}
    Our next assumption addresses the rank of the sample Jacobian of the function in the overparameterized regime.

    \begin{assumption}\label{A3}
        In the overparameterized regime $p \geq NK$, the sample Jacobian of the function at local minima $Z_S(\hvartheta_S)\in\bbR^{p\times NK}$ has full rank,
        \begin{align}                       \operatorname{rank}\left(Z_S(\hvartheta_S)\right)=NK.
        \end{align} 
    \end{assumption}

    Intuitively, Assumption~\ref{A3} means that the trainable parameters affect the $NK$ training outputs in sufficiently many independent ways.
    One may expect this condition to hold in sufficiently expressive PQCs, and we provide numerical evidence that it is satisfied in the settings considered here.
    We further assume that the spectrum of the uncentered (sample) covariance of function gradients remains stable during training.
    
    \begin{assumption}
        The smallest eigenvalue of the uncentered sample covariance of the function gradients remains comparable to its value at initialization uniformly in the number of trainable parameters $p$:
        \begin{align}               \lambda_{\min}\left(\hC_S^f(\hvartheta_S)\right) \asymp \lambda_{\min}\left(\hC_S^f(\vartheta_0)\right).
        \end{align}
        More precisely, there exist constants $A,B>0$ independent of $p$ such that
        \begin{align}
        A\,\lambda_{\min}\left(\hC_S^f(\vartheta_0)\right) \le \lambda_{\min}\left(\hC_S^f(\hvartheta_S)\right) \le B\,\lambda_{\min}\left(\hC_S^f(\vartheta_0)\right).
        \end{align}
        We further assume that an analogous relation holds for the uncentered covariance of function gradients $C_f$, with constants $A',B'>0$.
        These relations hold with high probability with respect to the random draw of the training set $S$, the initialization $\vartheta_0$, and the local minimum $\hvartheta_S$ returned by the training procedure.
    \label{A4}
    \end{assumption}

    Assumption~\ref{A4} formalizes a form of spectral stability during training. We provide empirical evidence that this assumption holds in the PQCs considered here.

    The next Assumption~\ref{A5} is a regularity condition on the covariance of function gradients at initialization.
    It ensures that the spectrum remains bounded and non-degenerate as the number of trainable parameters varies, thereby enabling random matrix theoretic analysis, as we discuss in more detail in Appendix~\ref{a:proofofmaintheorem}.

    \begin{assumption}[Spectrum of the uncentered covariance of function gradients at initialization]
        The uncentered covariance of function gradients at initialization has bounded spectrum.
        That is, there exist constants $c,c'>0$, independent of the number of trainable parameters $p$, such that
        \begin{align}    c\le\lambda_{\min}\left(C_f(\vartheta_0)\right)\le\lambda_{\max}\left(C_f(\vartheta_0)\right)\le c'.
        \end{align}
    \label{A5}
    \end{assumption}
    We expect this assumption to hold for circuit architectures with non-redundant parameterizations and generic initializations.
    Our next assumption relates the smallest non-zero eigenvalue of the empirical-risk Hessian to the spectra of its outer-product and functional contributions.

    \begin{assumption}[Weyl-like inequality]
        The smallest non-zero eigenvalue of the empirical-risk Hessian at $\hvartheta_S$ satisfies
        \begin{align}
            \lambda_{\min}(\hH_S^o(\hvartheta_S))+\lambda_{\min}(\hH_S^f(\hvartheta_S)) & \leq \lambda_{\min}(\hH_S(\hvartheta_S)) \leq \lambda_{\min}(\hH_S^o(\hvartheta_S))+\lambda_{\max}(\hH_S^f(\hvartheta_S)).
        \end{align}
    \label{A6}
    \end{assumption}
    Note that if $\operatorname{rank}(\hH_S^o(\hvartheta_S))=\operatorname{rank}(\hH_S(\hvartheta_S))$, then Weyl's inequality directly implies Assumption~\ref{A6}. One may therefore interpret this assumption as requiring that the functional Hessian does not alter the rank structure induced by the outer-product Hessian.
    The final assumption concerns the add-one-in loss term appearing in the expected-risk decomposition.
    
    \begin{assumption}[Add-one-in loss]
        The add-one-in loss $\calL'(S)$ does not dominate the lower bound in Theorem~\ref{thm:lower_bound_K1} in a neighborhood of the interpolation threshold.
    \label{A7}
    \end{assumption}
    Since $\calL'(S)$ corresponds to the loss obtained after perturbing the training set by a single sample, one may expect it to remain close to the empirical risk of the trained model, which is typically small around the interpolation threshold.

    We are now ready to present the main result of this work.
    The following theorem shows that, under the above assumptions, the lower bound on the expected risk from Theorem~\ref{thm:lower_bound_K1} attains a maximum at the interpolation threshold, leading to a double descent peak as the number of parameters increases.
    
    \begin{theorem}[Double descent peak at interpolation]\label{thm:peakatinterpolation}
        Let $S =\{(x_i, y_i)\}_{i=1}^N$ be a dataset with $d$-dimensional Gaussian normal inputs $x_i \sim \mathcal{N}(0, \mathbb{I}_d)$ and outputs $y_i \in \mathbb{R}$.
        Let $f_\vartheta : \mathbb{R}^d \to \mathbb{R}$ be a parameterized quantum circuit as in Eq.~\eqref{eq:PQCoutput} for $K=1$, consisting of $p$ parameters.
        Under Assumptions~\ref{A1}--\ref{A7}, and considering the mean-squared error loss function, the lower bound on $\calL(\hvartheta_S)$ at a local minimum from Theorem~\ref{thm:lower_bound_K1} attains a maximum at $p=N$ with high probability in the limit $N,p\to\infty$ with fixed ratio $p/N \to\gamma$.
    \end{theorem}
    \begin{proof}[Proof sketch]
        Starting from the lower bound of Theorem~\ref{thm:lower_bound_K1}, it suffices to analyze the dependence on $p$ of the Hessian-dependent contribution appearing in the bound.
        To this end, we use the decomposition from Eq.~\eqref{eq:Hessian_decomp}.
        Using the upper bound from Assumption~\ref{A6} together with Assumption~\ref{A4}, this term can be further bounded by the smallest non-zero eigenvalue of the uncentered sample covariance of function gradients at initialization.
        Consequently, we need to characterize the behavior of $\lambda_{\min}(\hC_S^f(\vartheta_0))$ as a function of the parameter dimension $p$.
        Under Assumption~\ref{A5} and using standard singular value inequalities, central results in random matrix theory dictate that $\lambda_{\min}(\hC_S^f(\vartheta_0))$ approaches zero at the interpolation threshold $p=N$, while remaining bounded away from zero on either side of this point with high probability.
        It therefore remains to control the second term in the denominator involving $\lambda_{\max}(\hH_S^f(\hvartheta_S))$.
        In the underparameterized regime, Assumption~\ref{A2} implies that $\lambda_{\max}(\hH_S^f(\hvartheta_S))$ decreases with increasing $p$, causing the lower bound to increase.
        In the overparameterized regime, we can show that Assumption~\ref{A3} implies $\hH_S^f(\hvartheta_S)=0$ at local minima.
        Combining these observations with the fact that $\lambda_{\min}(\hC_S^f(\vartheta_0))$ is minimized at the interpolation threshold, we conclude that the lower bound attains its maximum at $p=N$.       
    \end{proof}

    A corresponding result on the double descent peak holds for the upper bound derived in Appendix~\ref{app:upper}; see Theorem~\ref{thm:upper_bound_peak}, which shows that the upper bound likewise has a maximum at the interpolation threshold $p=N$.

    Theorem~\ref{thm:peakatinterpolation} is restricted to the scalar-output case $K=1$.
    Nevertheless, the interpolation threshold in the proof arises from the spectral properties of the sample Jacobian $Z_S(\vartheta)\in\mathbb{R}^{p\times NK}$, suggesting that the relevant transition should more generally occur when the number of parameters matches the number of scalar training constraints $NK$.
    This motivates the following conjecture.
    
    \begin{conjecture}[Interpolation threshold for $K>1$]
        Under the setting of Theorem~\ref{thm:peakatinterpolation}, except that $K>1$, the lower bound on the expected risk exhibits a double descent peak at $p=NK$.
    \end{conjecture} 

    Our numerical experiments provide evidence supporting this conjecture.

\section{Proof of Theorem~\ref{thm:peakatinterpolation}}\label{a:proofofmaintheorem}

    In this section, we prove Theorem~\ref{thm:peakatinterpolation}, which establishes that the lower bound on the expected risk from Theorem~\ref{thm:lower_bound_K1} attains its maximum at the interpolation threshold.
    We first collect several auxiliary results concerning the spectrum of the Hessian of the empirical risk and its decomposition into outer-product and functional contributions.
    These ingredients are then combined with random matrix arguments to characterize the behavior of the lower bound as a function of the parameter dimension.

    \begin{lemma}[Decomposition of the Hessian]\label{l:decomp_Hessian}
        For a given training set $S=\{z_i\coloneqq(x_i,y_i)\}_{i=1}^N$, we consider the empirical risk $\hcalL_S(\vartheta)=\frac{1}{N}\sum_{i=1}^N \ell(z_i,\vartheta)$.
        We define the Hessian of the empirical risk $\hH_S(\vartheta)\coloneqq\nabla_\vartheta\nabla_\vartheta^\intercal\hcalL_S(\vartheta)$.
        Then, the decomposition 
        \begin{align}
            \hH_S(\vartheta) &= \hH_S^o(\vartheta)+\hH_S^f(\vartheta) ,\\
            \text{outer-product Hessian: } \hH_S^o(\vartheta) &\coloneqq \frac{1}{N}\sum_{i=1}^N\nabla_\vartheta f_\vartheta(x_i)\left[\nabla_f\nabla_f^\intercal\ell(z_i,\vartheta)\right]\nabla_\vartheta^\intercal f_\vartheta(x_i) ,\\
            \text{functional Hessian: } \hH_S^f(\vartheta) &\coloneqq \frac{1}{N}\sum_{i=1}^N\nabla_f^\intercal\ell(z_i,\vartheta) \nabla_\vartheta\nabla_\vartheta^\intercal f_\vartheta(x_i)
        \end{align}
        holds.
    \end{lemma}
    \begin{proof}
        The proof follows the direct application of the chain rule, being careful about the sizes of the algebraic objects involved.
        To be precise, we treat $\nabla_\vartheta$ and $\nabla_\vartheta^\intercal$ as operator-valued column and row vectors, respectively; and similarly for $\nabla_f^{(\intercal)}$.
        
        For any $x$, we have $f_\vartheta(x)\in\bbR^K$, $\nabla_\vartheta^\intercal f_\vartheta(x)\in\bbR^{K\times p}$, and $\nabla_\vartheta\nabla_\vartheta^\intercal f_\vartheta(x)\in\bbR^{K\times(p\times p)}$.
        The second $\nabla_\vartheta$ acts along a different direction than the original dimension of the output.
        For any $z=(x,y)$, when instantiating the chain rule, we treat $\ell(z,\vartheta)$ as $\ell((f, y); \vartheta)$, with $f=f(x)$.
        Then, for any $z$, we write
        \begin{align}
            \underbrace{\nabla_\vartheta^\intercal\ell(z,\vartheta)}_{1\times p} &= \underbrace{\nabla_f^\intercal\ell(z,\vartheta)}_{1\times K}\underbrace{\nabla_\vartheta^\intercal f_\vartheta(x)}_{K\times p}.
        \end{align}
        Then, for the second derivative, we have
        \begin{align}
            \underbrace{\nabla_\vartheta\overbrace{\left(\nabla_f^\intercal\ell(z,\vartheta)\right)}^{1\times K}}_{p\times K} &= \underbrace{\nabla_\vartheta f_\vartheta(x)}_{p\times K}\underbrace{\nabla_f\nabla_f^\intercal\ell(z,\vartheta)}_{K\times K}.
        \end{align}
        We use this for the chain rule of the second derivative, to get
        \begin{align}
            \underbrace{\nabla_\vartheta\nabla_\vartheta^\intercal\ell(z,\vartheta)}_{p\times p} &= \nabla_\vartheta\left(\nabla_f^\intercal\ell(z,\vartheta)\nabla_\vartheta^\intercal f_\vartheta(x)\right) \\
            &= \left(\nabla_\vartheta\nabla_f^\intercal\ell(z,\vartheta)\right)\nabla_\vartheta^\intercal f_\vartheta(x) + \nabla_f^\intercal\ell(z,\vartheta)\left(\nabla_\vartheta \nabla_\vartheta^\intercal f_\vartheta(x)\right) \\
            &= \underbrace{\nabla_\vartheta f_\vartheta(x)}_{p\times K}\underbrace{\nabla_f\nabla_f^\intercal\ell(z,\vartheta)}_{K\times K}\underbrace{\nabla_\vartheta^\intercal f_\vartheta(x)}_{K\times p} + \underbrace{\nabla_f^\intercal\ell(z,\vartheta)}_{1\times K}\underbrace{\nabla_\vartheta \nabla_\vartheta^\intercal f_\vartheta(x)}_{K\times(p\times p)}.
        \end{align}
        Having proven this for arbitrary $z=(x,y)$, the linearity of the differential operators guarantees that it holds for each summand in $\hH_S(\vartheta)$, and hence for the whole sum.

        For completeness, we also give expressions for the entries of the two summands, avoiding the matrix-matrix product notation.
        Written out in terms of scalars, the $(j,j')^\text{th}$ entry of each of the summands in terms of the $k^\text{th}$ entry of the function $[f_\vartheta(x)]_k$ takes the form
        \begin{align}
            \left[\hH_S^o(\vartheta)\right]_{j,j'} &= \frac{1}{N}\sum_{i=1}^N \sum_{k,k'=1}^K \pder[{\left[f_\vartheta(x_i)\right]}_k]{\vartheta_j} 
            \pder[^2\ell(z_i,\vartheta)]{f_k\partial f_{k'}}\pder[{\left[f_\vartheta(x_i)\right]_{k'}}]{\vartheta_{j'}}, \\
            \left[\hH_S^f(\vartheta)\right]_{j,j'} &= \frac{1}{N}\sum_{i=1}^N \sum_{k=1}^K \pder[\ell(z_i,\vartheta)]{f_k}\pder[^2{\left[f_\vartheta(x_i)\right]_k}]{\vartheta_j\partial\vartheta_{j'}}.
        \end{align}
    \end{proof}

    \begin{lemma}[Theorem 2.1 in Ref.~\cite{yaskov2016necessary}]\label{l:yaskov}
        Let $X$ be a matrix with $N$ columns $x_p \in \mathbb{R}^p$. Assume $x_p$ are isotropic (i.e., $\mathbb{E}[x_px_p^\intercal]=\bbI_p$) random vectors and so-called concentration of quadratic forms 
        \begin{align}
            \frac{1}{p} \left[x_p^\intercal \left(\frac{1}{N}XX^\intercal+\varepsilon \bbI_p\right)^{-1} x_p - \Tr(\frac{1}{N}XX^\intercal+\varepsilon \bbI_p)^{-1}\right] \rightarrow 0
        \label{eq:concentration_quadratic_forms}
        \end{align}
        is satisfied, as $N \rightarrow \infty$ with $p=p(N)$ and $p/N \rightarrow \gamma$.
        Then, the empirical spectral distribution of the sample covariance $\frac{1}{N}XX^\intercal$ converges weakly to the \emph{Mar\v{c}enko-Pastur} (MP) distribution $\mu_\text{MP}$ in the asymptotic limit.
        Its non-zero eigenvalue density is given by
        \begin{align}
            \mu_{\mathrm{MP}}(d\lambda) &=\frac{\sqrt{(\lambda_+-\lambda)(\lambda-\lambda_-)}}{2\pi\gamma\lambda}d\lambda,\\
            \lambda_\pm &=(1\pm\sqrt{\gamma})^2,
        \end{align}
        where $\lambda$ denotes an eigenvalue of the sample covariance matrix
        $\frac{1}{N}XX^\intercal$, and $\lambda_-$ and $\lambda_+$ denote the lower and upper edges of the support of the non-zero spectrum, respectively.
    \end{lemma}

    \begin{definition}[{Exponentially-concentrated random vectors~\cite[Def.~2.1]{seddik2020random}}]
        Let $x\in\bbR^p$ a random vector $x\sim\calP$.
        We say $x$ is an \emph{exponentially-concentrated random vector} if, for any $1$-Lipschitz function $h:\bbR^p\to\bbR$, the following holds:
        \begin{align}
            \bbP_{x\sim\calP}\left[\left\lvert f(x) - \bbE_{x\sim\calP}\left[ f(x)\right]\right\rvert > t\right] &\leq A e^{-Bt^2},
        \end{align}
        for constants $A,B>0$.
    \end{definition}

    \begin{lemma}[Theorem 2.40 in Ref.~\cite{louart2021concentrationmeasurelargerandom}]\label{l:louart}
        Let $x\in\bbR^p$ be an exponentially-concentrated random vector $x\sim\calP$, and let $A\in\bbR^{p\times p}$ be a matrix with bounded operator norm.
        Then, $x^\intercal A x - \Tr\{A\}$ is sharply concentrated.
    \end{lemma}

    \begin{definition}[Data whitening]\label{def:whitening}
        Let us call $g_i = \nabla_\vartheta f_{\vartheta_0}(x_i)\in\bbR^p$ for each $i\in\{1,\ldots,N\}$, and recall that $C_f(\vartheta_0)\coloneqq\bbE_{z_i\sim\calD,\,\xi_z^2 >0}[g_ig_i^\intercal]$.
        Under Assumption~\ref{A5},  we refer to the following as \emph{whitened} random vectors: $\tg_i=C_f(\vartheta_0)^{-1/2}g_i$.
        At the matrix level, we define $\tZ_S(\vartheta_0)=(\tg_1,\ldots,\tg_N)\in\bbR^{p\times N}$, which corresponds to $\tZ_S(\vartheta_0)=C_f(\vartheta_0)^{-1/2}Z_S(\vartheta_0)$.
        Analogously, we define $\tC_S^f(\vartheta_0)=\tZ_S(\vartheta_0)\tZ_S(\vartheta_0)^\intercal/N$ (we drop the hat $\tilde \hC_S^f$ to avoid cluttering notation), which corresponds to $\tC_S^f(\vartheta_0) = C_f(\vartheta_0)^{-1/2} \hC_S^f(\vartheta_0) C_f(\vartheta_0)^{-1/2}$.
        We call $\tC_S^f(\vartheta_0)$ the \emph{whitened uncentered sample covariance of the function gradient at initialization}.
    \end{definition}

    \begin{lemma}[{Gaussian vectors ~\cite[Prop.~2.2]{seddik2020random}}] \label{l:gaussian_concentration}
    Let $x\sim\mathcal N(0,\bbI_d)$.
    Then $x$ is an exponentially-concentrated random vector.
    \end{lemma}

    \begin{lemma}[{Lipschitz stability of concentrated random vectors~\cite[Prop.~2.3]{seddik2020random}}]\label{l:lipschitz}
        Let $x\in\bbR^d$ be an exponentially-concentrated random vector, and let $G:\bbR^d\to\bbR^p$ be an $L$-Lipschitz function, where $L$ may depend on $p$.
        Then the random vector $G(x)$ is also exponentially concentrated.
    \end{lemma}

    \begin{lemma}[Properties of whitened vectors]\label{l:whitened_vectors}
        Under Assumption~\ref{A5}, let $f_{\vartheta}$ be a parameterized quantum circuit as in Eq.~\eqref{eq:PQCoutput} with $K=1$, and let $\tg_i=C_f(\vartheta_0)^{-1/2}g_i$ with $g_i=\nabla_\vartheta f_{\vartheta_0}(x_i)\in\bbR^p$ as in Def.~\ref{def:whitening}.
        For Gaussian normal inputs $x_i\sim\mathcal N(0,\bbI_d)$, the vectors $\tg_i$ are independent, isotropic, and fulfill concentration 
        of quadratic forms.
    \end{lemma}
    \begin{proof}
        Consider the notation from Def.~\ref{def:whitening}.
        Then, $Z_S(\vartheta_0) = \left(g_1,\ldots,g_N\right)\in\bbR^{p\times N}$, and still $\hC_S^f(\vartheta_0)=Z_S(\vartheta_0)Z_S(\vartheta_0)^\intercal/N$.
        Since the training data $x_i$ are sampled i.i.d.\ and $\vartheta_0$ is fixed at initialization, it follows that the random vectors $g_i$ are independent among themselves.
        This implies that $\tg_i$ are also independent.
        We immediately have that the whitened random vectors are isotropic $\bbE_{z_i\sim\calD}\left[\tg_i\tg_i^\intercal\right]=\bbI_p$, since
        \begin{align}
            \bbE_{z_i\sim\calD}\left[\tg_i\tg_i^\intercal\right] &= \bbE_{z_i\sim\calD}\left[C_f(\vartheta_0)^{-1/2} g_i g_i^\intercal C_f(\vartheta_0)^{-1/2}\right] = C_f(\vartheta_0)^{-1/2} \bbE_{z_i\sim\calD}\left[g_i g_i^\intercal\right]C_f(\vartheta_0)^{-1/2} \\
            &= C_f(\vartheta_0)^{-1/2} C_f(\vartheta_0)C_f(\vartheta_0)^{-1/2} = \bbI_p.
        \end{align}
        Here, we have used that $C_f(\vartheta_0)^{-1/2}=\left(C_f(\vartheta_0)^{-1/2}\right)^\intercal$ is symmetric and does not depend on $z_i$.

        Since further $x_i \sim \mathcal{N}(0, \mathbb{I}_d)$, we know by Lemma~\ref{l:gaussian_concentration} that they are concentrated random vectors.
        For standard PQCs, the map $x\mapsto \nabla_\vartheta f_{\vartheta_0}(x)$ is Lipschitz continuous, since the parameter-shift rule expresses gradient components as differences of Lipschitz-continuous function evaluations.
        It therefore follows from Lemma~\ref{l:lipschitz} that the function gradients $g_i=\nabla_\vartheta f_{\vartheta_0}(x_i)$ are also concentrated random vectors.
        Moreover, by construction, $\tg_i=C_f(\vartheta_0)^{-1/2}g_i$.
        Since $C_f(\vartheta_0)^{-1/2}$ is a linear map with bounded operator norm by Assumption~\ref{A5}, a second application of Lemma~\ref{l:lipschitz} shows that the whitened vectors $\tg_i$ are concentrated random vectors as well.
        Therefore, the concentration of quadratic forms follows from Lemma~\ref{l:louart}.
    \end{proof}

    \begin{lemma}[Whitened-vector matrix converges to MP law]\label{l:MP_whitened}
        In the setting of Lemma~\ref{l:whitened_vectors}, the empirical spectral distribution of the whitened uncentered sample covariance of the function gradient at initialization $\tC_S^f(\vartheta_0)$ converges weakly almost surely to the Mar\v{c}enko-Pastur law in the limit where $N, p \rightarrow \infty$ and $p/N \rightarrow \gamma$.
    \end{lemma}
    \begin{proof}
        Our strategy is to invoke Lemma~\ref{l:yaskov}, which requires two conditions: independent, isotropic columns and concentration of quadratic forms, which we confirm using Lemma~\ref{l:whitened_vectors}.
    \end{proof}

    \begin{corollary}[Smallest eigenvalue convergence for whitened vectors]\label{cor:smallest_eval_whitened}
        The smallest non-zero eigenvalue of the whitened uncentered sample covariance of the function gradient at initialization $\tC_S^f(\vartheta_0)$ converges almost surely to $(1-\sqrt{p/N})^2$.
    \end{corollary}
    \begin{proof}
        Follows directly from Lemma~\ref{l:MP_whitened}.
    \end{proof}

    \begin{lemma}[Smallest eigenvalue convergence for original vectors]\label{l:smallest_eval}
        In the setting of Lemma~\ref{l:whitened_vectors}, the smallest non-zero eigenvalue of the (original) sample covariance of the function gradient at initialization $\hC_S^f(\vartheta_0)$ achieves a minimum at $N=p$, in the limit where $N,p\to\infty$ and $p/N \rightarrow \gamma$.
    \end{lemma}
    \begin{proof}
        We relate the eigenvalues of both sample covariance matrices: the whitened one and the original one.
        In both cases, the smallest non-zero eigenvalues of the sample covariance matrix $\lambda_{\min}(\tC_S^f(\vartheta_0))$ correspond to the smallest non-zero singular values of the matrix of function gradients $\sigma_{\min}^2(\tZ_S(\vartheta_0))$ (and analogously for $\hC_S^f(\vartheta_0)$ and $Z_S(\vartheta_0)$).
        Using the identity $\tZ_S(\vartheta_0)=C_f(\vartheta_0)^{-1/2}Z_S(\vartheta_0)$, standard singular value inequalities (see, e.g., Ref.~\cite{mathias2013singular}) imply
        \begin{align}        \sigma_{\min}\left(\tZ_S(\vartheta_0)\right)\sigma_{\min}\left(C_f(\vartheta_0)^{1/2}\right) &\leq \sigma_{\min}\left(Z_S(\vartheta_0)\right) \leq \sigma_{\min}\left(\tZ_S(\vartheta_0)\right)\sigma_{\max}\left(C_f(\vartheta_0)^{1/2}\right).
        \end{align}
        By Corollary~\ref{cor:smallest_eval_whitened}, $\sigma_{\min}\left(\tZ_S(\vartheta_0)\right)$ converges to $1-\sqrt{p/N}$ with high probability, which is $\sigma_{\min}\left(\tZ_S(\vartheta_0)\right)\to0$ in the stated limit of $p=N$.
        From Assumption~\ref{A5} it then follows that$\sigma_{\min}\left(Z_S(\vartheta_0)\right)\to0$ in the same limit, thus achieving a minimum.
    \end{proof}
    
    \begin{lemma}[Adapted from Ref.~\cite{karhadkar2024mildlyoverparameterized}]\label{l:global_optimum}
        Let $K=1$, $p\geq N$, and $\hvartheta_S$ a local minimum as in Definition~\ref{def:localminimum}.
        Under Assumption~\ref{A3}, and considering the mean-squared error loss function, $\hvartheta_S$ is a global minimum of the empirical risk \begin{equation}
        \hvartheta_S\in\arg\min_{\vartheta}\left[\hcalL_S(\vartheta)\right].\end{equation}
    \end{lemma}
    \begin{proof}
        We start with the definition of local minima: $\nabla_\vartheta\hcalL_S(\vartheta) = 0$.
        From the definition of the empirical risk, this becomes $\sum_{i=1}^N \nabla_\vartheta\ell(z_i,\hvartheta_S)=0$.
        We first use the chain rule $\nabla_\vartheta\ell(z_i,\hvartheta_S)=\nabla_\vartheta f_{\hvartheta_S}(x_i)\nabla_f\ell(z_i,\hvartheta_S)$, and we note that, for the mean-squared error loss $\ell(z,\vartheta)=(f_\vartheta(x)-y)^2/2$, we have $\nabla_f\ell(z,\vartheta)=f_\vartheta(x)-y$.
        We introduce the vector of errors $\hE_S=(f_{\hvartheta_S}(x_i)-y_i)_{i=1}^N$, and we re-write the sum over $x_i\in S$ as a matrix-vector multiplication, in order to get
        \begin{align}
            \sum_{i=1}^N \nabla_\vartheta\ell(z_i,\hvartheta_S) &= \sum_{i=1}^N \nabla_\vartheta f_{\hvartheta_S}(x_i)\nabla_f\ell(z_i,\hvartheta_S)
            = \sum_{i=1}^N \nabla_\vartheta f_{\hvartheta_S}(x_i)(f_\vartheta(x_i)-y_i)
            = Z_S(\hvartheta_S) \hE_S.
        \end{align}
        With this, the condition that $\hvartheta_S$ is a local minimum becomes $Z_S(\hvartheta_S)\hE_S=0$.

        Assumption~\ref{A3} dictates that, in the overparameterized regime $p\geq N$, the sample Jacobian matrix fulfills $\operatorname{rank}(Z_S(\hvartheta_S))=N$ (is full rank).
        From this, it follows that the only possibility for $Z_S(\hvartheta_S)\hE_S=0$ to hold is that $\hE_S=0$.
        Since $\hE_S$ is the vector of errors, it being zero means that the model makes no errors on the training set, which by definition implies the empirical risk achieves its lowest possible value $0$.
        This confirms that $\hvartheta_S$ is a global minimum of $\hcalL_S$.
    \end{proof}

    \begin{customtheorem}{\ref{thm:peakatinterpolation}}[Double descent peak at interpolation]
        Let $S =\{(x_i, y_i)\}_{i=1}^N$ be a dataset with $d$-dimensional Gaussian normal inputs $x_i \sim \mathcal{N}(0, \mathbb{I}_d)$ and outputs $y_i \in \mathbb{R}$.
        Let $f_\vartheta : \mathbb{R}^d \to \mathbb{R}$ be a parameterized quantum circuit as in Eq.~\eqref{eq:PQCoutput} for $K=1$, consisting of $p$ parameters.
        Under Assumptions~\ref{A1}--\ref{A7}, and considering the mean-squared error loss function, the lower bound on $\calL(\hvartheta_S)$ at a local minimum from Theorem~\ref{thm:lower_bound_K1} attains a maximum at $p=N$ with high probability in the limit $N,p\to\infty$ with fixed ratio $p/N \to\gamma$.
    \end{customtheorem}
    \begin{proof}
        Let us start from the lower bound in Theorem~\ref{thm:lower_bound_K1}, given by
        \begin{align}
            \calL(\hvartheta_S) &\geq \calL'(S) + \frac{1}{N+1}\frac{\alpha\xi_{\min}^2\lambda_{\min}(C_f(\hvartheta_S))}{\lambda_{\min}\left(\hH_S(\hvartheta_S)\right)} + \calO\left(N^{-2}\right).
        \end{align}
        We first turn our attention to the decomposition of $\hH_S(\hvartheta_S)$ from Lemma~\ref{l:decomp_Hessian}: $\hH_S(\hvartheta_S) = \hH_S^o(\hvartheta_S)+\hH_S^f(\hvartheta_S)$.
        Then, we apply the upper bound from Assumption~\ref{A6}: $\lambda_{\min}\left(\hH_S(\hvartheta_S)\right) \leq \lambda_{\min}\left(\hH_S^o(\hvartheta_S)\right)+\lambda_{\max}\left(\hH_S^f(\hvartheta_S)\right)$.
        Also, specializing to the mean-squared error loss simplifies the expression for the outer-product Hessian, since the first derivative yields $\nabla_f\ell(z_i,\vartheta)=(f_\vartheta(x_i)-y_i)$ and the second derivative $\nabla_f(f_\vartheta(x_i)-y_i))=1$.
        With these, we have
        \begin{align}
            \hH_S^o(\vartheta) &\coloneqq \frac{1}{N}\sum_{i=1}^N\nabla_\vartheta f_\vartheta(x_i)\underbrace{\left[\nabla_f\nabla_f^\intercal\ell(z_i,\vartheta)\right]}_{=1}\nabla_\vartheta f_\vartheta(x_i)^\intercal 
            = \frac{1}{N}\sum_{i=1}^N\nabla_\vartheta f_\vartheta(x_i)\nabla_\vartheta f_\vartheta(x_i)^\intercal \eqcolon \hC_S^f(\vartheta).
        \end{align}
        The outer-product Hessian becomes the uncentered sample covariance of the function gradient.
        Substituting both this and the inequality from Assumption~\ref{A6} into the lower bound, we obtain
        \begin{align}
            \calL(\hvartheta_S) &\geq \calL'(S) + \frac{1}{N+1}\frac{\alpha\xi_{\min}^2\lambda_{\min}(C_f(\hvartheta_S))}{\lambda_{\min}\left(\hC_S^f(\hvartheta_S)\right)+\lambda_{\max}\left(\hH_S^f(\hvartheta_S)\right)} + \calO\left(N^{-2}\right).
        \end{align}
        With Assumption~\ref{A4}, we can further replace the $\hvartheta_S$ dependence by the value of the parameters at initialization $\vartheta_0$, for some constants $A,B>0$, to get
        \begin{align}
            \calL(\hvartheta_S) &\geq \calL'(S) + \frac{1}{N+1}\frac{\alpha\xi_{\min}^2A\lambda_{\min}(C_f(\vartheta_0))}{B\lambda_{\min}\left(\hC_S^f(\vartheta_0)\right)+\lambda_{\max}\left(\hH_S^f(\hvartheta_S)\right)} + \calO\left(N^{-2}\right).
        \end{align}
        
        First of all, recall that the first term in both regimes, $\calL'(S)=\bbE_{z'\sim\calD}\left[\ell(z',\hvartheta_{S'})\right]$, denotes the add-one-in loss. Since it depends on a training set perturbed by only a single sample, we do not expect its variation with $p$ to dominate the Hessian-dependent complexity term near interpolation.

        Let us now consider the overparameterized regime $p\geq N$.
        From Assumption~\ref{A3} and Lemma~\ref{l:global_optimum}, for our choice of mean-squared error loss, we have that $\hH_S^f(\hvartheta_S)=0$.
        The denominator in the lower bound then becomes simply $B\lambda_{\min}\left(\hC_S^f(\hvartheta_S)\right)$.
        Wielding Lemma~\ref{l:smallest_eval} we note that $\lambda_{\min}\left(\hC_S^f(\hvartheta_S)\right)$ achieves a minimum at $p=N$, and hence it is an increasing function of $p$ for $p\geq N$.
        With this, we confirm that the lower bound of $\calL(\hvartheta_S)$ is decreasing in the overparameterized regime.

        In the underparameterized regime $p<N$, the local minimum $\hvartheta_S$ is not necessarily a global minimum of the empirical risk, and hence both summands in the denominator of the lower bound remain.
        In this case, we combine Assumption~\ref{A2} and Lemma~\ref{l:smallest_eval} to argue that, with $p<N$, both summands are decreasing functions of $p$.
        For $\lambda_{\max}\left(\hH_S^f(\hvartheta_S)\right)$, Assumption~\ref{A2} explicitly states this fact; and for $\lambda_{\min}\left(\hC_S^f(\vartheta_0)\right)$, it follows $\lambda_{\min}\left(\hC_S^f(\vartheta_0)\right)$ having a minimum at $p=N$.
        We thus confirm that the lower bound of $\calL(\hvartheta_S)$ is increasing in the underparameterized regime near interpolation.
        The last two paragraphs combined complete the proof: the lower bound of $\calL(\hvartheta_S)$ has a maximum at $p=N$.
    \end{proof}

\section{Upper bound}\label{app:upper}

    In this section we complement Theorems~\ref{thm:lower_bound_K1} and~\ref{thm:peakatinterpolation} with analogous results for an \emph{upper bound} on the expected risk.
    We use the same notation and language as in Appendices~\ref{a:expected_risk_local_minima} and~\ref{a:dd_behavior_in_PQCs}.

    \begin{theorem}[Upper bound on the expected risk]\label{thm:upper_bound_K1}
        Under the setting of Lemma~\ref{l:error_decomp}, under Assumption~\ref{A1} and for $K=1$, the expected risk $\calL(\hvartheta_S)$ at the local minimum $\hvartheta_S$ fulfills the following upper bound:
        \begin{align}
            \calL(\hvartheta_S) &\leq\calL'(S) + \frac{1}{N+1}\frac{\alpha\,\xi^2_{\max}\,R\,\lambda_{\max}(C_f(\hvartheta_S))}{\lambda_{\min}(\hH_S(\hvartheta_S))} + \calO\left(\frac{1}{N^2}\right).
        \end{align}
        Here, we define $\xi^2_{\max} = \max_{\substack{z\sim\calD \\ \xi_z^2>0}} \{\xi_z^2\}$, $R=\operatorname{rank}(\hH_S(\hvartheta_S))$, $\lambda_{\min}(\cdot)$ the smallest non-zero eigenvalue of its matrix argument, and analogously for the largest eigenvalue $\lambda_{\max}(\cdot)$.
    \end{theorem}
    \begin{proof}
        The proof follows the steps of the proof of Theorem~\ref{thm:lower_bound_K1}, using the opposite direction of Lemma~\ref{l:eval_trace_ineq},
        \begin{align}
            \calL(\hvartheta_S) &=\calL'(S) + \frac{1}{N+1}\Tr\left\{\left[\hH_S(\hvartheta_S)\right]^{-1}C(\hvartheta_S)\right\}+\calO\left(\frac{1}{N^2}\right) \\
            &\leq \calL'(S) + \frac{1}{N+1}\left(\alpha\,\xi_{\max}^2 \Tr\left\{\left[\hH_S(\hvartheta_S)\right]^{-1} C_f(\hvartheta_S)\right\}\right)+\calO\left(\frac{1}{N^2}\right) \\
            &\leq \calL'(S) + \frac{\alpha\,\xi_{\max}^2}{N+1}\left( \lambda_{\max}(C_f(\hvartheta_S))\Tr\left\{\left[\hH_S(\hvartheta_S)\right]^{-1}\right\}\right)+\calO\left(\frac{1}{N^2}\right) \\
            &\leq \calL'(S) + \frac{\alpha\,\xi_{\max}^2\lambda_{\max}(C_f(\hvartheta_S))}{N+1}\left( \frac{R}{\lambda_{\min}(\hH_S(\hvartheta_S))}\right)+\calO\left(\frac{1}{N^2}\right) \\
            &=\calL'(S) + \frac{1}{N+1}\frac{\alpha\,\xi^2_{\max}\,R\,\lambda_{\max}(C_f(\hvartheta_S))}{\lambda_{\min}(\hH_S(\hvartheta_S))} + \calO\left(\frac{1}{N^2}\right).
        \end{align}
        The only qualitative difference in this expression is in the trace of the inverse of the Hessian.
        This time, we upper bound a sum of positive terms by the largest of the terms $(\lambda_{\min}(\hH_S(\hvartheta_S)))^{-1}$ times the number of terms $R$.
        This concludes the proof.
    \end{proof}

    \begin{customassumption}{2'}\label{A2'}
        In the underparametrized regime $p<NK$, the smallest eigenvalue $\lambda_{\min}(\hH_S^f(\hvartheta_S))$ of the functional Hessian at local minima decreases with the number of trainable parameters $p$ in a neighborhood of the interpolation threshold $p=NK$.
    \end{customassumption}

    \begin{customassumption}{4'}\label{A4'}        
        The largest eigenvalue of the uncentered covariance of the function gradients remains comparable to its value at initialization, uniformly in the number of parameters $p$, $\lambda_{\max}(C_f(\hvartheta_S))\asymp\lambda_{\max}(C_f(\vartheta_0))$.
        This relation holds with high probability with respect to the random draw of the training set $S$, the initialization $\hvartheta_0$, and the local minimum $\hvartheta_S$ returned by the training procedure.
    \end{customassumption}

    \begin{theorem}[Upper bound peak at interpolation]\label{thm:upper_bound_peak}
        Let $S =\{(x_i, y_i)\}_{i=1}^N$ be a dataset with $d$-dimensional Gaussian normal inputs $x_i \sim \mathcal{N}(0, \mathbb{I}_d)$ and outputs $y_i \in \mathbb{R}$.
        Let $f_\theta : \mathbb{R}^d \to \mathbb{R}$ be a parameterized quantum circuit as in Eq.~\eqref{eq:PQCoutput} for $K=1$, consisting of $p$ parameters.
        Under Assumptions~\ref{A1}, \ref{A2'}, \ref{A3}, \ref{A4}, \ref{A4'}, and \ref{A5}--\ref{A7}, and considering the mean-squared error loss function, the upper bound on $\calL(\hvartheta_S)$ at a local minimum from Theorem~\ref{thm:upper_bound_K1} attains a maximum at $p=N$ with high probability in the limit $N,p\to\infty$ with fixed ratio $p/N\to\gamma$.
    \end{theorem}
    \begin{proof}
        The proof follows the steps of the proof of Theorem~\ref{thm:peakatinterpolation}.
        We start from the upper bound in Theorem~\ref{thm:upper_bound_K1}, which is
        \begin{align}
            \calL(\hvartheta_S) &\leq\calL'(S) + \frac{1}{N+1}\frac{\alpha\,\xi^2_{\max}\,R\,\lambda_{\max}(C_f(\hvartheta_S))}{\lambda_{\min}(\hH_S(\hvartheta_S))} + \calO\left(\frac{1}{N^2}\right).
        \end{align}
        We first turn our attention to the decomposition of $\hH_S(\hvartheta_S)$ from Lemma~\ref{l:decomp_Hessian}: $\hH_S(\hvartheta_S) = \hH_S^o(\hvartheta_S)+\hH_S^f(\hvartheta_S)$.
        Then, we apply the other direction of the inequality in Assumption~\ref{A6}, which is
        \begin{equation}\lambda_{\min}\left(\hH_S(\hvartheta_S)\right) \geq \lambda_{\min}\left(\hH_S^o(\hvartheta_S)\right)+\lambda_{\min}\left(\hH_S^f(\hvartheta_S)\right).
        \end{equation}
        Again, specializing to the mean-squared error loss simplifies the expression for the outer-product Hessian, we have $\hH_S^o(\vartheta) = \hC_S^f(\vartheta)$.
        Substituting both this and the inequality from Assumption~\ref{A6} into the upper bound, we obtain
        \begin{align}
            \calL(\hvartheta_S) &\leq \calL'(S) + \frac{1}{N+1}\frac{\alpha\xi_{\max}^2 R\lambda_{\max}(C_f(\hvartheta_S))}{\lambda_{\min}\left(\hC_S^f(\hvartheta_S)\right)+\lambda_{\min}\left(\hH_S^f(\hvartheta_S)\right)} + \calO\left(N^{-2}\right).
        \end{align}
        With Assumptions~\ref{A4} and~\ref{A4'}, we can further replace the $\hvartheta_S$ dependence by the value of the parameters at initialization $\vartheta_0$, for some constants $A,B>0$,
        \begin{align}
            \calL(\hvartheta_S) &\leq \calL'(S) + \frac{1}{N+1}\frac{\alpha\xi_{\max}^2 RB\lambda_{\max}(C_f(\vartheta_0))}{A\lambda_{\min}\left(\hC_S^f(\vartheta_0)\right)+\lambda_{\min}\left(\hH_S^f(\hvartheta_S)\right)} + \calO\left(N^{-2}\right).
        \end{align}
        
        Let us now consider the overparameterized regime $p\geq N$.
        From Assumption~\ref{A3} and Lemma~\ref{l:global_optimum}, for our choice of mean-squared error loss, we have that $\hH_S^f(\hvartheta_S)=0$.
        The denominator in the lower bound then becomes simply $A\lambda_{\min}\left(\hC_S^f(\hvartheta_S)\right)$.
        Wielding Lemma~\ref{l:smallest_eval} we note that $\lambda_{\min}\left(\hC_S^f(\hvartheta_S)\right)$ achieves a minimum at $p=N$, and hence it is an increasing function of $p$ for $p\geq N$.
        With this, we confirm that the upper bound of $\calL(\hvartheta_S)$ is decreasing in the overparameterized regime.

        In the underparameterized regime $p<N$, the local minimum $\hvartheta_S$ is not necessarily a global minimum of the empirical risk, and hence both summands in the denominator of the upper bound remain.
        In this case, we combine Assumption~\ref{A2'} and Lemma~\ref{l:smallest_eval} to argue that, with $p<N$, both summands are decreasing functions of $p$.
        For $\lambda_{\min}\left(\hH_S^f(\hvartheta_S)\right)$, Assumption~\ref{A2'} explicitly states this fact; and for $\lambda_{\min}\left(\hC_S^f(\vartheta_0)\right)$, it follows from $\lambda_{\min}\left(\hC_S^f(\vartheta_0)\right)$ having a minimum at $p=N$.
        We thus confirm that the upper bound of $\calL(\hvartheta_S)$ is increasing in the underparameterized regime near interpolation. 
        The last two paragraphs combined complete the proof: the upper bound of $\calL(\hvartheta_S)$ has a maximum at $p=N$.
    \end{proof}

\section{Further empirical evidence of double descent in PQCs}\label{a:evidence}

    \begin{figure*}[t]
            \centering
            \includegraphics[width=0.97\textwidth]{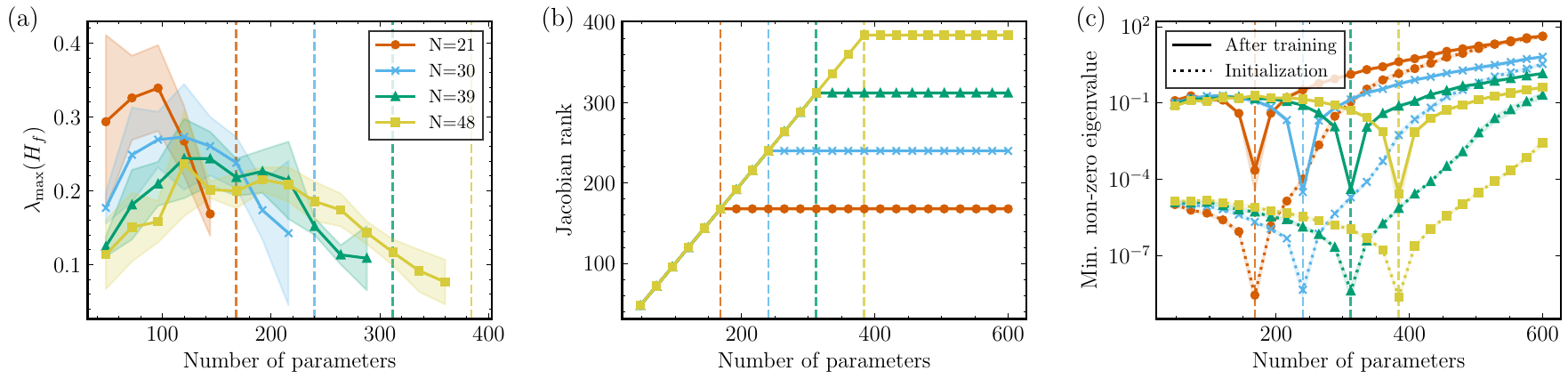} 
            \caption{
                Numerical verification of the assumptions underlying the double descent analysis for the MNIST-1D dataset, for different training set sizes $N$. Vertical dashed 
                lines indicate the predicted interpolation thresholds $p=NK$.
                (a) Largest eigenvalue of the functional Hessian $\hH_S^f(\hvartheta_S)$ after training.
                In the underparameterized regime $p<NK$, $\lambda_{\max}(\hH_S^f(\hvartheta_S))$ decreases as $p$ approaches interpolation, supporting Assumption~\ref{A2}.
                (b) Rank of the sample Jacobian of the function $Z_S(\hvartheta_S)$ after training.
                For $p\geq NK$, the rank saturates at $NK$, supporting Assumption~\ref{A3}.
                (c) Smallest non-zero eigenvalue of the uncentered sample covariance matrix of function gradients $\hC_S^f$ at initialization and after training, supporting Assumption~\ref{A4}.
                The shaded areas correspond to the standard deviation for ten independent experiment repetitions, each using independently sampled
                training data.
            }
            \label{fig:dd_assumptions}
    \end{figure*}
    
    Complimenting our main numerical findings reported in Fig.~\ref{fig:dd_main}, we now turn our attention to a selection of the assumptions we introduced for Theorem~\ref{thm:peakatinterpolation} in Appendix~\ref{a:dd_behavior_in_PQCs}.
    We present our results in Fig.~\ref{fig:dd_assumptions}, where we probe further the mechanism behind the double descent peak.
    Here, we focus on the MNIST-1D dataset.
    By systematically testing for different assumptions, we strengthen the evidence that the peak at interpolation in Fig.~\ref{fig:dd_main} is accurately captured by our main analytical contribution, Theorem~\ref{thm:peakatinterpolation}.
    \begin{itemize}
        \item In Fig.~\ref{fig:dd_assumptions}(a) we test Assumption~\ref{A2}: in the underparameterized regime $p<NK$, the largest eigenvalue of the functional Hessian, $\lambda_{\max}(\hH_S^f(\hvartheta_S))$, decreases as the number of parameters approaches the interpolation threshold.
        This indicates that the contribution of the functional Hessian becomes less dominant near interpolation, as required in the proof of the double descent peak.
        \item In Fig.~\ref{fig:dd_assumptions}(b) we test Assumption~\ref{A3}: we show that, once $p\geq NK$, the rank of the sample Jacobian of the function $Z_S(\hvartheta_S)$ saturates at $NK$ for all training-set sizes considered.
        Thus, in the overparameterized regime, the Jacobian has full rank with respect to the $NK$ scalar training constraints.
        \item In Fig.~\ref{fig:dd_assumptions}(c) we test Assumption~\ref{A4}: we compare the smallest non-zero eigenvalue of the uncentered sample covariance matrix of function gradients $\hC_S^f$ at initialization and after training, displaying a dip near the interpolation threshold.
    \end{itemize}
    
\section{Further details on the numerical experiments}\label{app:numerics}

    In this section, we provide further implementation details for the numerical experiments presented in Figs.~\ref{fig:dd_main} and~\ref{fig:dd_assumptions}. All quantum circuit simulations were performed with the \texttt{PennyLane} software library~\cite{bergholm2018pennylane}, and all experiments use $n=8$ qubits and $K=8$ output dimensions. The circuit consists of an initial angle embedding of the classical input, followed by $L$ data re-uploading layers.
    Each layer contains trainable single-qubit rotations $R_X$, $R_Z$, and $R_Y$ on every qubit, followed by nearest-neighbor $CZ$ gates arranged in a ring. The input is then re-uploaded before the next trainable layer.
    The model output is obtained by measuring Pauli-$X$ expectation values on the first $K$ qubits and multiplying the resulting vector by the fixed scale factor $c=150$.
    Since each layer contains three trainable parameters per qubit, the total number of trainable parameters is
        $p = 3nL$.
    In the experiments, the depth is varied as $L=2,\ldots,25$, corresponding to $p=48,\ldots,600$ trainable parameters. All models are trained with the mean squared error loss using the Adam optimizer, and the training is performed for $2500$ epochs.
    
    For the MNIST-1D and Fashion MNIST classification tasks, we use $8$ classes.
    The labels are encoded as one-hot vectors and subsequently mapped from $\{0,1\}$ to $\{-1,1\}$, so that the correct class has value $+1$ and all other classes have value $-1$.
    This allows the classification tasks to be trained using the same mean squared error objective as in the theoretical analysis.
    Classical inputs are reduced to dimension $d=n=8$ using principal component analysis, implemented with \texttt{scikit-learn}~\cite{pedregosa2011scikit}, and then rescaled component-wise to the interval $[-\pi/2,\pi/2]$ before being passed to the angle embedding.
    For the synthetic regression task, inputs are sampled uniformly from $[-\pi/2,\pi/2]^8$, and targets are generated from a multidimensional linear model $y = X W + \zeta$, where $W$ is the all-ones matrix in $\mathbb{R}^{8\times 8}$ and $\zeta$ is Gaussian noise with standard deviation $0.5$.
    
    To scan across the interpolation threshold, we vary the number of training samples as
        $N \in \{21,30,39,48\}$.
    Since $K=8$, the corresponding interpolation thresholds are $p=NK\in\{168,240,312,384\}$.
    For each value of $N$ and each circuit depth, the test loss is evaluated on $1000$ test samples.
    These values are chosen such that the interpolation thresholds coincide exactly with parameter counts attainable by varying the number of data re-uploading layers.
    The results are averaged over $10$ independent repetitions, and shaded regions in the figures indicate the standard deviation over these repetitions. The code is publicly available in Ref.~\cite{github_repository}.

\section{Universality classes of double descent}\label{app:universality}

The proof of Theorem \ref{thm:peakatinterpolation}  establishes the double descent peak by relating the smallest non-zero eigenvalue of the empirical covariance matrix of function gradients to the soft edge of the Mar\v{c}enko–Pastur law. This naturally raises the question of how much of the phenomenon depends on the specific Wishart ensemble arising from independent Gaussian data, and how much is instead a manifestation of a broader random-matrix universality principle.

Viewed abstractly, the proof relies on only two ingredients. First, the smallest non-zero singular value of the empirical Jacobian develops a soft edge at the interpolation threshold. Second, the associated singular vectors are sufficiently delocalized so that generic data directions overlap with them in a statistically uniform way. The specific Gaussian assumptions merely provide one convenient realization of these properties.

A natural question is to what extent the present analysis depends
on the specific Wishart ensemble underlying the
Mar\v{c}enko--Pastur law.
From the perspective of the proof, the essential ingredient is not
Gaussianity itself but rather the emergence of a soft spectral edge
governing the smallest non-zero singular value of the empirical
covariance matrix.
This suggests that analogous double descent behavior may persist
for considerably broader classes of covariance ensembles.
Prominent examples include correlated Wishart ensembles,
random kernel matrices, and deformed covariance ensembles
\cite{DoubleDescent3,DoubleDescent4,DoubleDescent5}.
More generally, nonlinear kernel models whose limiting spectral
measures differ from the classical Mar\v{c}enko--Pastur law may
provide a mathematically natural framework for extending the
present analysis beyond linearized parameterizations
\cite{DoubleDescent2,DoubleDescent4}.
These observations suggest that the mechanism underlying
Theorem~\ref{thm:peakatinterpolation} should be understood as a manifestation of a broader
universality class rather than of a single random matrix model.

\begin{conjecture}[Universality of double descent]
Consider the setting of Theorem~\ref{thm:peakatinterpolation}, replacing the Gaussian Wishart covariance model by a random covariance ensemble whose smallest non-zero singular value exhibits universal soft-edge behavior and whose singular vectors remain asymptotically delocalized. Then the lower bound on the expected risk exhibits an interpolation peak at the corresponding critical aspect ratio, giving rise to double descent behavior.
\end{conjecture}

Establishing the precise universality class governing double descent in parameterized quantum circuits, and identifying which structural properties of the underlying random matrix ensemble are genuinely required, constitute interesting open problems.

\end{document}